\documentclass[a4paper,onecolumn,11pt,accepted=2026-10-19]{quantumarticle}
\pdfoutput=1
\usepackage{graphicx,color}
\usepackage{amsmath}
\usepackage{amssymb}
\usepackage{amsthm}
\usepackage{revsymb}
\usepackage{bm}
\usepackage{hyperref}
\usepackage{cite}
\usepackage{enumitem}  
\newtheorem{thm}{Theorem}
\newtheorem{prop}[thm]{Proposition}
\newtheorem{lemma}[thm]{Lemma}
\newtheorem{corol}[thm]{Corollary}

\theoremstyle{definition}
\newtheorem{definition}[thm]{Definition}
\newtheorem{example}{Example}

\newtheorem{remark}{Remark}

\numberwithin{equation}{section}

\newcommand{\rmd}{\mathrm{d}}
\newcommand{\rmi}{\mathrm{i}}
\newcommand{\rme}{\mathrm{e}}

\newcommand{\Tr}{\operatorname{Tr}}
\renewcommand{\Re}{\operatorname{Re}}
\renewcommand{\Im}{\operatorname{Im}}
\newcommand{\rank}{\operatorname{rank}}
\newcommand{\diag}{\operatorname{diag}}

\newcommand{\ad}{\operatorname{ad}}

\newcommand{\T}{\operatorname{T}}

\newcommand{\ket}[1]{|{#1}\rangle}
\newcommand{\bra}[1]{\langle{#1}|}

\newcommand{\bracket}[2]{\langle#1|#2\rangle}

\newcommand{\ii}{\mathrm{i}}
\newcommand{\ee}{\mathrm{e}}
\newcommand{\dd}{\mathrm{d}}

\definecolor{dgreen}{rgb}{0,0.5,0}

\definecolor{dblue}{rgb}{0,0,0.6}

\definecolor{dred}{rgb}{0.784,0,0}

\definecolor{dorange}{cmyk}{0,0.72,1,0.16}
\definecolor{dmagenta}{rgb}{0.847,0.149,0.490}

\definecolor{delete}{cmyk}{0.5,0,0,0}

\begin{document}
\title[RWA for Open Systems]{Rotating-Wave and Secular Approximations for Open Quantum Systems }
% \date{February 12, 2026}
\author{Daniel Burgarth}
%\email{daniel.burgarth@fau.de}
\orcid{0000-0003-4063-1264}
\affiliation{Department Physik, Friedrich-Alexander-Universit\"at Erlangen-N\"urnberg, Staudtstra\ss e 7, 91058 Erlangen, Germany}
\author{Paolo Facchi}
%\email{paolo.facchi@ba.infn.it}
\orcid{0000-0001-9152-6515}
\affiliation{Dipartimento di Fisica, Universit\`a di Bari, I-70126 Bari, Italy}
\affiliation{INFN, Sezione di Bari, I-70126 Bari, Italy}
\author{Giovanni Gramegna}
%\email{giovanni.gramegna@uniba.it}
\orcid{0000-0001-7532-1704}
\affiliation{Dipartimento di Fisica, Universit\`a di Bari, I-70126 Bari, Italy}
\affiliation{INFN, Sezione di Bari, I-70126 Bari, Italy}
\author{Kazuya Yuasa}
%\email{yuasa@waseda.jp}
\orcid{0000-0001-5314-2780}
\affiliation{Department of Physics, Waseda University, Tokyo 169-8555, Japan}
%%%%%%%%%%%%%%%%%%%%%%%%%%%%%%%%%%%
\maketitle
\begin{abstract}
 We derive a nonperturbative bound on the distance between evolutions of open quantum systems described by time-dependent generators. 
 We show how this result can be employed to provide an explicit upper bound on the error of the rotating-wave approximation in the presence of dissipation and decoherence. 
 We apply the derived bound to the strong-coupling limit in open quantum systems and to the secular approximation used to obtain a master equation of the Gorini-Kossakowski-Lindblad-Sudarshan form from the Redfield master equation.
\end{abstract}
%\tableofcontents

\section{Introduction}
The rotating-wave approximation (RWA) is one of the most widely used tools in quantum physics, underpinning simplified and effective descriptions in quantum optics, condensed-matter physics, and quantum information science~\cite{bloch1940magnetic,shirley1965solution,haeberlen1968coherent,Agarwal,Nazir}. By neglecting rapidly oscillating terms in an interaction picture, the RWA allows one to replace a complicated time-dependent dynamics with a tractable effective generator, often capturing the essential long-time behavior of the system. Despite its ubiquity and practical success, the approximation is typically justified through heuristic arguments or perturbative reasoning, and rigorous quantitative error bounds are comparatively scarce.

Recently, explicit bounds for closed systems have been developed for a variety of models, often based on norm estimates or integration-by-parts techniques~\cite{burgarth2022one,burgarth2024taming,richter2024quantifying,heib2025bounding,dey2025error}. These results clarify the role of the separation of timescales and show how fast oscillations suppress certain contributions to the evolution. However, in many physical contexts where one would apply the RWA, the  relevant systems are intrinsically open, and hence it is necessary to extend such results to quantum dynamical evolutions described by possibly time-dependent generators of the Gorini-Kossakowski-Lindblad-Sudarshan (GKLS) form~\cite{lindblad1976generators,gorini1976completely,FloquetLindbladian}. Such an extension presents both conceptual and technical challenges.

From a conceptual point of view, it is natural to ask how the dissipative part of the generator of the evolution is affected when the RWA is applied to the Hamiltonian part of the generator~\cite{band2015open,fleming2010rotating}. Naively, one might apply the RWA to the Hamiltonian part and subsequently model the dissipation using a GKLS generator with fixed decay rates, without explicitly accounting for the effect of the transformation on the noise terms. However, a rigorous approach requires one to deal with both the Hamiltonian  and the noise part of the evolution with a suitable rotating reference frame, which might affect the noise~\cite{fleming2010rotating,Kohler}. This also raises the question whether the resulting effective generator still has the GKLS form. 
In this regard, one of the paradigmatic examples in which the RWA is performed is the microscopic derivation of a GKLS generator from the Redfield equation, which is not in the GKLS form. Nevertheless, in some settings one works directly with the Redfield equation instead of its GKLS counterpart, despite the possibility of observing unphysical effects~\cite{Hartmann,benatti2022local,Trushechkin}. It is therefore natural to ask how close the resulting dynamics are, and under what conditions the two descriptions meaningfully differ.
Furthermore, one may ask whether the RWA can still be justified when the rapidly oscillating (``rotating'') components of the generator include dissipation as well, rather than only the Hamiltonian term.

From the technical point of view, the proof strategy used in the unitary case hinges on explicitly inverting the evolution. While the evolution of an open quantum system is still mathematically invertible (as a linear map), its inverse is in general not contractive. In fact, the norm of the inverse map may grow with the relevant limiting parameter, which requires a careful choice of the operators to isolate in order to obtain a useful bound. Moreover, many techniques in the closed-system scenario revolve around the spectral representation of the generators, while generators of completely positive trace-preserving (CPTP) evolutions are not guaranteed to be diagonalizable.

In this work, we address these questions by deriving nonperturbative bounds, general structure theorems, and illustrating examples, furthermore unifying several known results in open quantum systems. As a main application, we use these bounds to obtain an explicit error estimate for the RWA in the presence of dissipation. This allows us to clarify how the dissipative part of the generator should be treated within the approximation and to quantify the regime in which the RWA remains valid. Moreover, our results provide an alternative route to strong-coupling limits and shed light on the secular approximation commonly used in the derivation of GKLS master equations. 

Our error bounds depend on the duration of the evolution and do not, in general, establish uniform accuracy for arbitrarily long times. For the secular approximation, however, an improved estimate allows the time interval over which a prescribed accuracy can be guaranteed to grow linearly with the parameter controlling the separation of timescales, with the remaining parameters held fixed.

\subsection{Summary of the Results and Outline of the Paper}
The rest of the paper is structured as follows.

In Section~\ref{sec:Preliminaries}, we introduce notation and basic concepts concerning evolution operators generated by time-dependent GKLS generators.

In Section~\ref{sec:IntPartLemma}, we provide the main technical tool of the paper, an integration-by-parts lemma (Lemma~\ref{lem:IntParts}) that underlies our error estimates. The essential improvement provided by this result with respect to the standard Duhamel formula (Lemma~\ref{lem:GroupDiff}) is the introduction of a reference frame $\Lambda_0(t,s)$ (not necessarily unitary) and the isolation of the integral action $\mathcal{S}_{12}(t)$ in this reference frame.

In Section~\ref{sec:Application}, we apply this framework to derive bounds when the generator contains a strong part responsible for the rapidly oscillating terms. Theorem~\ref{thm:ConstL0} is the main theorem of the paper and contains two error bounds: one on the error of the RWA~\eqref{eqn:BoundThm2}, and the other on the distance between the true evolution and the evolution projected on the peripheral subspace~\eqref{eqn:BoundThm2P}. We stress that in Theorem~\ref{thm:ConstL0} we do not assume the generator to have the GKLS form. In Theorems~\ref{rmk:PhysicalBound} and~\ref{rmk:PhysicalBoundUnitary}, we provide improved versions of the bounds which are much tighter if the generators of the evolutions are endowed with additional structures, such as contractivity, GKLS form, and unitarity. Corollary~\ref{cor:ConstL01} specializes Theorem~\ref{thm:ConstL0} to the RWA, and Corollary~\ref{cor:ConstL02} takes into account the possibility of having different timescales. In addition, Corollary~\ref{cor:StrongLimit} shows how the framework can be used to obtain a known result (the strong-coupling limit).

In Section~\ref{sec:Examples}, we illustrate our bounds with concrete examples and compare them with numerics.

Finally, in Section~\ref{sec:SecApp}, we apply the main result to provide an explicit bound on the distance between the Redfield evolution and the CPTP evolution obtained in the secular approximation.

The technical results used to derive the bounds in the diamond norm are collected in Appendix~\ref{app:DiamondNorm}, while some elementary facts on superoperator algebra which are used in the examples are recalled in Appendix~\ref{app:ElementaryFacts}\@.

\section{Preliminaries on Evolution Operators}
\label{sec:Preliminaries}
Let us consider a norm-continuous time-dependent family $t \mapsto \mathcal{L}(t)$ of bounded operators on a Banach space, and denote by $\Lambda(t,s)$ the evolution operator they generate from time $s$ to time $t$ with $0\leq s\leq t$, which is the solution of the evolution equations
\begin{equation}\label{eq:evolution}
	\frac{\partial}{\partial t}\Lambda(t,s)
	=\mathcal{L}(t)\Lambda(t,s),\qquad
	\frac{\partial}{\partial s}\Lambda(t,s)
	=-\Lambda(t,s)\mathcal{L}(s),\qquad
	\Lambda(s,s)=1.
\end{equation}
We will often use the following notations, setting one of the two time arguments of $\Lambda(t,s)$ to $0$,
\begin{equation}
	\Lambda(t)=\Lambda(t,0),\qquad
	\Lambda^{-1}(t)=\Lambda(0,t).
\end{equation}
In order to derive our bounds, we will use the $L^\infty$ and $L^1$ norms, defined as follows. 
Given a continuous family $u\mapsto \mathcal{A}(u)$ of bounded operators on a Banach space, with $u\in[s,t]$, we define
\begin{equation}
\|\mathcal{A}\|_{\infty,[s,t]}:=\sup_{u\in[s,t]}\|\mathcal{A}(u)\|,\qquad
\|\mathcal{A}\|_{1,[s,t]}:=\int_s^t\rmd u\,\|\mathcal{A}(u)\|,
\end{equation}
for $t\ge s$, where the norm $\|\mathcal{A}(u)\|$ can be a general operator norm.
Among various operator norms, we will in particular use the diamond norm $\|\mathcal{A}(u)\|_\diamond$ for maps describing the evolutions of quantum systems, whose properties are recalled in Appendix~\ref{app:DiamondNorm}\@.
We will also use the simplified notations
\begin{equation}
\|\mathcal{A}\|_{\infty,t}=\|\mathcal{A}\|_{\infty,[0,t]},\qquad
\|\mathcal{A}\|_{1,t}=\|\mathcal{A}\|_{1,[0,t]},
\end{equation}
for $t\ge0$.

We are in particular interested in the physical evolutions of open quantum systems.
They are described by completely positive and trace-preserving (CPTP) maps~\cite{alicki2007quantum,nielsen2001quantum,chruscinski2017brief}\@.
The framework we are going to establish, however, does not really need the CPTP properties, but the uniform boundedness of evolution is enough.
Therefore, we will often focus on families $\{\mathcal{L}(t)\}_{0\leq t\leq T}$ of generators of contraction semigroups, satisfying $\|\rme^{s\mathcal{L}(t)}\|\le1$ for all $s\ge0$ and $t\in[0,T]$.
\begin{prop}[Ref.~\cite{reed1975ii}, Theorem~X.70\@. See also Refs.~\cite{kato1953integration,avron2012adiabatic}]\label{prop:Contraction}
	Let $t\in[0,T]\mapsto \mathcal{L}(t)$ be a continuous family of bounded generators of contraction semigroups. Then, the time propagator generated by $\mathcal{L}(t)$ according to~\eqref{eq:evolution} is a contraction, i.e.~$\|\Lambda(t,s)\|\le1$ for $0\le s\le t\le T$. 
\end{prop}
\begin{remark}
\label{rmk:Diamond}
If $\mathcal{L}(t)$ generates a CPTP map $\Lambda(t,s)$ for $t\ge s$, its diamond norm is $\|\Lambda(t,s)\|_\diamond=1$ for $t\ge s$. See Ref.~\cite{watrous2018theory} and Appendix~\ref{app:DiamondNorm}\@.
\end{remark}

A characterization of the generators $\mathcal{L}(t)$ of CPTP evolutions is given in the following.
\begin{prop}[Ref.~\cite{chruscinski2022dynamical}, Corollary~7]
	Let $\Lambda(t,s)$ be a family of evolution operators satisfying~\eqref{eq:evolution} for all $0\leq s\leq t\leq T$. Then, $\Lambda(t,s)$ is a CPTP map for all $0\leq s\leq t\leq T$, if and only if $\{\mathcal{L}(t)\}_{t\in[0,T]}$ are generators of the GKLS form, i.e.
	\begin{equation}
		\mathcal{L}(t)\varrho =-\rmi [H(t),\varrho]-\frac{1}{2}\sum_{k}\gamma_k(t)\,\Bigl(
		\{V_k(t)^\dagger V_k(t),\varrho\}
		-2V_k(t)\varrho V_k(t)^\dagger
		\Bigr),
	\end{equation}
	where $t\rightarrow H(t)$ is a family of self-adjoint operators, $\gamma_k(t)\geq 0$ for all $k$ and $t\in[0,T]$, and $\{V_k(t)\}$ are called jump operators.
\end{prop}

The time propagators defined in~\eqref{eq:evolution} can be expressed by the Dyson series
\begin{equation}
	\Lambda(t,s)
	=\T\exp\!\left(
	\int_s^t\rmd u\,\mathcal{L}(u)
	\right)
	=1+\sum_{n=1}^\infty\int_s^t\rmd u_1\cdots\int_s^{u_{n-1}}\rmd u_n\,
	\mathcal{L}(u_1)
	\cdots
	\mathcal{L}(u_n),
	\label{eqn:Dyson}
\end{equation}
whose convergence is guaranteed by bounding each term of the series as 
\begingroup
\allowdisplaybreaks
\begin{align}
\|\Lambda(t,s)\|
&\le1+\sum_{n=1}^\infty\int_s^t\rmd u_1\cdots\int_s^{u_{n-1}}\rmd u_n\,
\|\mathcal{L}\|_{\infty,[s,t]}^n
\nonumber\\
&=
1+\sum_{n=1}^\infty\frac{1}{n!}(t-s)^n
\|\mathcal{L}\|_{\infty,[s,t]}^n
=\rme^{(t-s)\|\mathcal{L}\|_{\infty,[s,t]}}.
\label{eqn:DysonBound}
\end{align}
\endgroup
In many situations, the bound~\eqref{eqn:DysonBound} is too loose. 
For example, by virtue of Proposition~\ref{prop:Contraction}, the evolution $\Lambda(t,s)$ generated by the generator $\mathcal{L}(u)$ of a contraction semigroup for $u\in[0,T]$ is actually bounded by $\|\Lambda(t,s)\|\le1$ for $0\le s\le t\le T$.

On the other hand, if the (possibly unbounded) generator $\mathcal{L}_0(u)$ of a contraction semigroup for $u\in[0,T]$ is perturbed by a bounded $\mathcal{D}(u)$, the evolution $\Lambda(t,s)$ generated by $\mathcal{L}(t)=\mathcal{L}_0(t)+\mathcal{D}(t)$ can be given by another Dyson series
\begin{align}
\Lambda(t,s)
=
\Lambda_0(t,s)
+\sum_{n=1}^\infty\int_s^t\rmd u_1
\int_s^{u_1}\rmd u_2
\cdots\int_s^{u_{n-1}}\rmd u_n\,
&
\Lambda_0(t,u_1)
\mathcal{D}(u_1)
\Lambda_0(u_1,u_2)
\mathcal{D}(u_2)
\cdots
\nonumber\\
&\qquad
{}\times
\Lambda_0(u_{n-1},u_n)
\mathcal{D}(u_n)
\Lambda_0(u_n,s),
\label{eqn:Dyson2}
\end{align}
where $\Lambda_0(t,s)=\T\rme^{\int_s^t\rmd u\,\mathcal{L}_0(u)}$. 
Here, the unperturbed evolution is a contraction bounded by $\|\Lambda_0(t,s)\|\le1$ for $0\le s\le t\le T$, and $\Lambda(t,s)$ can be bounded by the norm of the perturbation as
\begin{align}
\|\Lambda(t,s)\|
&\le
1
+\sum_{n=1}^\infty
\int_s^t\rmd u_1\cdots\int_s^{u_{n-1}}\rmd u_n\,
\|\mathcal{D}\|_{\infty,[s,t]}^n
\nonumber\\
&=
1
+\sum_{n=1}^\infty\frac{1}{n!}(t-s)^n
\|\mathcal{D}\|_{\infty,[s,t]}^n
=
\rme^{(t-s)\|\mathcal{D}\|_{\infty,[s,t]}},
\label{eqn:DysonBound2}
\end{align}
for $0\le s\le t\le T$.

\section{Integration-by-Parts Lemma}
\label{sec:IntPartLemma}
The objective of this paper is to provide a framework that allows us to prove the RWA for open quantum systems.
In order to estimate the error of an approximation, we compare the approximate evolution with the true evolution.
We can compare the two evolutions generated by two time-dependent generators $\mathcal{L}_1(t)$ and $\mathcal{L}_2(t)$ by the following elementary lemma.
\begin{lemma}
\label{lem:GroupDiff}
Consider two continuous time-dependent bounded operators $t\mapsto \mathcal{L}_j(t)$ ($j=1,2$), and the time propagators generated by them,
\begin{equation}
\Lambda_j(t,s)
=\T\exp\!\left(
\int_s^t\rmd u\,\mathcal{L}_j(u)
\right)
\qquad(j=1,2).
\end{equation}
One has
\begin{equation}
\Lambda_1(t,s)-\Lambda_2(t,s)
=
\int_s^t\rmd u\,\Lambda_1(t,u)[\mathcal{L}_1(u)-\mathcal{L}_2(u)]\Lambda_2(u,s).
\label{eqn:GroupDiff}
\end{equation}
Moreover, if $\mathcal{L}_j(t)$ is the generator of a contraction semigroup for each $j=1,2$ and all $t\in[0,T]$, the distance between the evolutions is uniformly bounded as
\begin{equation}
	\|\Lambda_1(t)-\Lambda_2(t)\|
\leq
\|
	\mathcal{L}_1-\mathcal{L}_2\|_{1,T},
\label{eq:boundPropDiff}
\end{equation}
for each $t\in[0,T]$.
\end{lemma}
\begin{proof}
The difference between the two propagators $\Lambda_1(t,s)$ and $\Lambda_2(t,s)$ can be arranged as
\begingroup
\allowdisplaybreaks
\begin{align}
\Lambda_1(t,s)
-
\Lambda_2(t,s)
&=
-\Lambda_1(t,u)\Lambda_2(u,s)\biggr|_{u=s}^{u=t}
\nonumber\\
&=
-\int_s^t\rmd u\,\frac{\partial}{\partial u}[\Lambda_1(t,u)\Lambda_2(u,s)]
\nonumber\\
&
=\int_s^t\rmd u\,\Lambda_1(t,u)[\mathcal{L}_1(u)-\mathcal{L}_2(u)]\Lambda_2(u,s).
\end{align}
\endgroup
This is~\eqref{eqn:GroupDiff}. If $\mathcal{L}_j(t)$ generates a contraction semigroup for each $j=1,2$ and all $t\in[0,T]$, then $\|\Lambda_{j}(t,s)\|\leq 1$ for each $0\leq s\leq t \leq T$, and by taking the norm, \eqref{eq:boundPropDiff} follows.
\end{proof}

This lemma shows that the distance between the two evolutions $\Lambda_1(t,s)$ and $\Lambda_2(t,s)$ is small if the two generators $\mathcal{L}_1(t)$ and $\mathcal{L}_2(t)$ are close to each other.
What is nontrivial in the RWA, however, is the fact that the effective generator in the RWA is not really close to the original generator, but still the effective generator yields an evolution that is close to the true evolution.
Lemma~\ref{lem:GroupDiff} is not useful to prove the RWA\@.
Our key instrument, on the other hand, is the following integration-by-parts lemma.
The basic idea is to go to an appropriate reference frame and to average a rapidly oscillating generator over time to get an effective generator.
The following integration-by-parts lemma allows us to implement this idea, and to estimate the error of the RWA\@.
In contrast to the unitary case~\cite{burgarth2022one}, we need to be careful with the irreversibility of the evolutions in dealing with open quantum systems.
The evolutions of open quantum systems are uniformly bounded forward in time but are not backward.
We devise the instrument taking care of this fact.

\begin{lemma}[Integration-by-parts lemma]
\label{lem:IntParts}
Consider three continuous time-dependent bound\-ed generators $t \mapsto \mathcal{L}_j(t)$ ($j=0,1,2$), and the propagators they generate,
\begin{equation}
\Lambda_j(t,s)
=\T\exp\!\left(
\int_s^t\rmd u\,\mathcal{L}_j(u)
\right)\qquad
(j=0,1,2),
\end{equation}
where $j=0$ plays the role of a reference propagator, while $j=1,2$ refer to the propagators to be compared.
Let the  propagator $\Lambda_2$ be split as
\begin{equation}
 \Lambda_2(t)
=\Lambda_0(t)\widetilde{\Lambda}_2(t),
\label{eq:divisible}
\end{equation}
where
\begin{equation}
 \widetilde{\Lambda}_2(t)
=\T\exp\!\left(\int_0^t\rmd s\,\widetilde{\mathcal{L}}_2(s)\right),
\qquad
\widetilde{\mathcal{L}}_2(t)  = \Lambda_0(t)^{-1}[\mathcal{L}_2(t)-\mathcal{L}_0(t)]\Lambda_0(t).
\label{eq:TildeLambda2}
\end{equation}
Define the integral action with respect to the reference propagator $\Lambda_0$ as
\begin{equation}
\mathcal{S}_{12}(t)=\int_0^t\rmd s\,\Lambda_0(t,s)[\mathcal{L}_1(s)-\mathcal{L}_2(s)]\Lambda_0(s).
\label{eqn:Action}
\end{equation}

Then, one has
\begin{equation}
\Lambda_1(t)
-
\Lambda_2(t)
=\mathcal{S}_{12}(t)\widetilde{\Lambda}_2(t)
+\int_0^t\rmd s\,\Lambda_1(t,s)\,\Bigl(
[\mathcal{L}_1(s)-\mathcal{L}_0(s)]\mathcal{S}_{12}(s)
-\mathcal{S}_{12}(s)\widetilde{\mathcal{L}}_2(s)
\Bigr)\,{\widetilde{\Lambda}_2(s)},
\label{eqn:KeyFormalism}
\end{equation}
and the following bound holds for $t\ge0$,
\begin{equation}
\|\Lambda_1(t)-\Lambda_2(t)\|
\le
\|\mathcal{S}_{12}\|_{\infty,t}
\|\widetilde{\Lambda}_2\|_{\infty,t}
\,\biggl[
1
+
\|\Lambda_1(t,{}\cdot{})\|_{\infty,t}\,
\Bigl(
\|\mathcal{L}_1-\mathcal{L}_0\|_{1,t}
+
\|\widetilde{\mathcal{L}}_2\|_{1,t}
\Bigr)
\biggr].
\label{eqn:KeyFormalismBound}
\end{equation}
\end{lemma}
\begin{proof}
Using Lemma~\ref{lem:GroupDiff} and the splitting~\eqref{eq:divisible}, one has
\begin{align}
\Lambda_1(t)
-
\Lambda_2(t)
&=\int_0^t\rmd s\,\Lambda_1(t,s)[\mathcal{L}_1(s)-\mathcal{L}_2(s)]\Lambda_2(s) 
\nonumber\\
&=\int_0^t\rmd s\,\Lambda_1(t,s)[\mathcal{L}_1(s)-\mathcal{L}_2(s)]\Lambda_0(s)\widetilde{\Lambda}_2(s).
\label{eqn:GroupDiff2}
\end{align}
In order to extract $\mathcal{S}_{12}(t)$, first note that
\begin{equation}\label{eq:DS}
	\frac{\dd}{\dd s}\mathcal{S}_{12}(s)=\mathcal{L}_0(s)\mathcal{S}_{12}(s)+[\mathcal{L}_1(s)-\mathcal{L}_2(s)]\Lambda_0(s).
\end{equation}
Using~\eqref{eq:DS} in~\eqref{eqn:GroupDiff2} and performing an integration by parts, one gets
\begin{align}
\Lambda_1(t)
-
\Lambda_2(t)
={}&\int_0^t\rmd s\,\Lambda_1(t,s)\left(\frac{\dd}{\dd s} \mathcal{S}_{12}(s) - \mathcal{L}_0(s)\mathcal{S}_{12}(s) \right)\widetilde{\Lambda}_2(s)
\nonumber\\
={}&
{\mathcal{S}_{12}(t)}{\widetilde{\Lambda}_2(t)}
+\int_0^t\rmd s\,{\Lambda_1(t,s)}[\mathcal{L}_1(s)-\mathcal{L}_0(s)]\mathcal{S}_{12}(s)\widetilde{\Lambda}_2(s)
\nonumber\\
&
{}-\int_0^t\rmd s\,\Lambda_1(t,s)\mathcal{S}_{12}(s)\widetilde{\mathcal{L}}_2(s)\widetilde{\Lambda}_2(s),
\end{align}
which is~\eqref{eqn:KeyFormalism}.
By triangle inequality, one can bound it for $t\ge0$ as
\begingroup
\allowdisplaybreaks
\begin{align}
\|
\Lambda_1(t)
-
\Lambda_2(t)
\|
\le{}&\|\mathcal{S}_{12}(t)\|\|\widetilde{\Lambda}_2(t)\|
\nonumber\\
&{}
+\int_0^t\rmd s\,\|\Lambda_1(t,s)\|\|\mathcal{S}_{12}(s)\|\|\widetilde{\Lambda}_2(s)\|
\,\Bigl(
\|\mathcal{L}_1(s)-\mathcal{L}_0(s)\|
+\|\widetilde{\mathcal{L}}_2(s)\|
\Bigr),
\end{align}
\endgroup
and get~\eqref{eqn:KeyFormalismBound}.
\end{proof}

In the following, Lemma~\ref{lem:IntParts} will be used to approximate the evolution generated by $\mathcal{L}_1(t)= \mathcal{L}_\kappa(t)$ in the limit of some control parameter $\kappa$. In order to do this, the task is to find $\mathcal{L}_{2}(t)=\mathcal{L}_{\mathrm{eff},\kappa}(t)$ generating an effective evolution $\Lambda_2(t)=\Lambda_{\mathrm{eff},\kappa}(t)$  and a suitable reference frame $\Lambda_0(t)=\Lambda_{0,\kappa}(t)$ such that $\mathcal{S}_{12}(t)=\mathcal{S}_\kappa(t)\rightarrow 0$. Then, Lemma~\ref{lem:IntParts} can be used to prove $\Lambda_\kappa(t)-\Lambda_{\mathrm{eff},\kappa}(t)\to0$, under boundedness conditions on $\|\Lambda_\kappa(t,s)\|$, $
\|\mathcal{L}_\kappa(t)-\mathcal{L}_{0,\kappa}(t)\|$, $\|\tilde{\mathcal{L}}_{\mathrm{eff},\kappa}(t)\|$, and $\|\tilde{\Lambda}_{\mathrm{eff},\kappa}(t)\|$, where $\mathcal{L}_0(t)=\mathcal{L}_{0,\kappa}(t)$ is the generator of the reference evolution $\Lambda_{0,\kappa}(t)$.
The basic idea is to use the fact that the time-average of a function of time rapidly oscillating around zero becomes small in the limit of high frequency.
To this end, in the unitary case~\cite{burgarth2022one}, we go to the interaction picture (rotating frame) with respect to the strong part of the Hamiltonian to get a highly oscillating Hamiltonian on the rotating frame, and integrate it to get a small integral action.
This helps us to show that the distance between the true and effective evolutions becomes small in the limit of some control parameter $\kappa$ and to prove various limit theorems, including the RWA, adiabatic theorems, and product formulas~\cite{burgarth2022one}.
We basically do the same for the nonunitary case.
However, we need to be careful with the interaction picture in the nonunitary case, since the generator of the evolution $\widetilde{\Lambda}_{\mathrm{eff},\kappa}(t)$
in the rotating frame with respect to $\Lambda_0(t)$ is  given by $\widetilde{\mathcal{L}}_{\mathrm{eff},\kappa}(t)=\Lambda_{0,\kappa}^{-1}(t)[\mathcal{L}_{\mathrm{eff},\kappa}(t)-\mathcal{L}_{0,\kappa}(t)]\Lambda_{0,\kappa}(t)$, and is not guaranteed to be bounded uniformly in the control parameter $\kappa$ because of the inverse $\Lambda_{0,\kappa}^{-1}(t)$ of a generally irreversible evolution.
That is why we define the integral action $\mathcal{S}_{12}(t)$ as~\eqref{eqn:Action}, putting $\Lambda_0(t)$ in front of the generators in the interaction picture to turn $\Lambda_0^{-1}(s)$ into $\Lambda_0(t,s)$, which is bounded for $t\ge s$ for contraction semigroups.
See Proposition~\ref{prop:Contraction} and Remark~\ref{rmk:Diamond} above.
Keeping these points in mind, we use Lemma~\ref{lem:IntParts} to find a limit evolution $\Lambda_2(t)$ of $\Lambda_1(t)$ in the limit of some control parameter $\kappa$.
In the limit of rapid oscillations and/or strong decay of $\Lambda_0(t)$, we get $\mathcal{S}_{12}(s)\to0$.

In this paper, we focus on the proof of the RWA for open quantum systems\@.
To this end, we analyze the time-dependent generator $\mathcal{L}_\kappa(t)$ whose strong part $\kappa\mathcal{L}_0$ is constant.
If we use Lemma~\ref{lem:IntParts} to analyze the generator $\mathcal{L}_\kappa(t)$ whose strong part $\kappa\mathcal{L}_0(t)$ is time-dependent, we end up with an adiabatic theorem.
The application to the adiabatic theorem will be presented in a sequel to this paper.

\section{Main Result: Time-Independent Reference Generator}
\label{sec:Application}
Let us consider a generator of the form
\begin{equation}
\mathcal{L}_\kappa(t)
=\kappa\mathcal{L}_0+\mathcal{D}_\kappa(t)
\end{equation}
on a finite-dimensional Banach space, consisting of a constant strong part $\kappa\mathcal{L}_0$ and a continuous time-dependent perturbation $\mathcal{D}_\kappa(t)$. 
The parameter $\kappa$ should be understood as a scale-separation parameter at the level of reduced dynamics: it rescales the oscillation frequencies associated with the peripheral eigenvalues of $\mathcal{L}_0$, as well as the decay rates associated with its nonperipheral spectrum.
We assume that $\mathcal{L}_0$ is the generator of a contraction semigroup, satisfying 
\begin{equation}
\|\rme^{t\mathcal{L}_0}\|\le 1,
\end{equation}
for all $t\ge0$.
We also assume that $\mathcal{D}_\kappa(t)$ is uniformly bounded,
\begin{equation}
\|\mathcal{D}_\kappa(t)\|\le D,
\end{equation}
for all $t\in[0,T]$ and $\kappa>0$, with some $D\ge0$.
We wish to find the limit evolution of
\begin{equation}
\Lambda_\kappa(t)
=\T\exp\!\left(
\int_0^t\rmd s\,\mathcal{L}_\kappa(s)
\right),
\end{equation}
in the limit $\kappa\to+\infty$, for $t\in[0,T]$.

The contractivity of the evolution $\rme^{t\mathcal{L}_0}$ ensures that the propagator $\Lambda_\kappa(t,s)$
is bounded by
\begin{equation}
\|\Lambda_\kappa(t,s)\|
\le
\rme^{(t-s)D}
\le\rme^{DT},
\label{eqn:BoundedEvolution}
\end{equation}
for $0\le s\le t\le T$ and $\kappa>0$, see~\eqref{eqn:DysonBound2}.
Let
\begin{equation}
\mathcal{L}_0
=
\sum_k(\alpha_k\mathcal{P}_k+\mathcal{N}_k)
\label{eqn:SpectralRep}
\end{equation}
be the spectral representation of $\mathcal{L}_0$~\cite{kato2013perturbation}, where $\{\alpha_k\}$ is the spectrum of $\mathcal{L}_0$, $\{\mathcal{P}_k\}$ and $\{\mathcal{N}_k\}$ are the spectral projections and the nilpotents, respectively. They satisfy 
\begin{equation}\label{eq:spectralProps}
	\mathcal{P}_k\mathcal{P}_{\ell}=\delta_{k\ell}\mathcal{P}_k, \quad\sum_{k}\mathcal{P}_k=\mathrm{1}, \qquad\text{and}\qquad \mathcal{N}_k\mathcal{P}_\ell=\mathcal{P}_\ell\mathcal{N}_k=\delta_{k\ell}\mathcal{N}_k, \quad \mathcal{N}_k^{\rho_k}=0,
\end{equation}
where the integer  $0<\rho_k\leq \rank \mathcal{P}_k$ is the degree of the nilpotent $\mathcal{N}_k$.
If $\mathcal{L}_0$ is the generator of a contraction semigroup, its spectrum $\{\alpha_k\}$ is confined in the left-half plane, $\Re\alpha_k\le0$, and the purely imaginary eigenvalues are semisimple, with no nilpotents $\mathcal{N}_k=0$ for $\alpha_k\in\rmi\mathbb{R}$ (for a proof, see~\cite[Proposition~6.2]{wolf2020quantum} or~\cite[Lemma~A.1]{hasenohrl2022interaction}).
In the following analysis, the peripheral projection of $\mathcal{L}_0$, defined by 
\begin{equation}
\mathcal{P}_\varphi=\sum_{\alpha_k\in\rmi\mathbb{R}}\mathcal{P}_k
\end{equation}
will play an important role.
It is also a contraction and we have~\cite[Proposition~6.3]{wolf2020quantum}\cite{wolf2008dividing,wolf2010inverse}
\begin{equation}\label{eq:PeripheralContraction}
\|\rme^{t\mathcal{L}_0}\mathcal{P}_\varphi\|{{}=\|\mathcal{P}_\varphi\|}\le 1,
\end{equation}
for $t\in\mathbb{R}$, so that $\rme^{t\mathcal{L}_0}$ restricted to the range of its peripheral projection $\mathcal{P}_\varphi$ is an isometry.
Note that, since $\|\mathcal{P}_\varphi\|=\|\mathcal{P}_\varphi^2\|\le\|\mathcal{P}_\varphi\|^2$, if $\mathcal{P}_\varphi\neq 0$, then $\|\mathcal{P}_\varphi\|\geq1$, and the inequality~\eqref{eq:PeripheralContraction} for $t=0$ implies that $\|\mathcal{P}_\varphi\|=1$.
The decaying part of $\rme^{t\mathcal{L}_0}$, on the other hand, is bounded by~\cite{burgarth2019generalized}
\begin{equation}
\|\rme^{t\mathcal{L}_0}\mathcal{Q}_\varphi\|\le\rme^{-\gamma_0t}p(t),
\label{eqn:NonperipheralDecay}
\end{equation}
for $t\ge0$, where $\mathcal{Q}_\varphi=1-\mathcal{P}_\varphi$ is the projection onto the nonperipheral spectrum of $\mathcal{L}_0$,  
\begin{equation}
\gamma_0=\min_{\alpha_k\not\in\rmi\mathbb{R}}|{\Re\alpha_k}|
\label{eqn:eta}
\end{equation}
is the smallest nonzero decay rate of $\mathcal{L}_0$ (we set $\gamma_0=\infty$ if the spectrum of $\mathcal{L}_0$ is completely peripheral), and
\begin{equation}
\label{eq:polynil}
p(t)=\sum_{\alpha_k\not\in\rmi\mathbb{R}}\|\mathcal{P}_k\|\sum_{n=0}^{\rho_k-1}\frac{1}{n!}\|\mathcal{N}_k\|^nt^n
\end{equation} 
is a polynomial with nonnegative coefficients, $\rho_k$ being the degree of nilpotent $\mathcal{N}_k$ defined in~\eqref{eq:spectralProps}.
We will use the bound
\begin{equation}
\|\rme^{s\mathcal{L}_0}\mathcal{Q}_\varphi\|_{1,\infty}
=
\int_0^\infty\rmd s\,
\|\rme^{s\mathcal{L}_0}\mathcal{Q}_\varphi\|
\le
\int_0^\infty\rmd s\,
\rme^{-\gamma_0s}p(s)
=
\frac{1}{\gamma_0}q(1/\gamma_0)
=
\frac{1}{\gamma_0}R,
\label{eqn:R}
\end{equation}
with $q(t)=\sum_{\alpha_k\not\in\rmi\mathbb{R}}\|\mathcal{P}_k\|\sum_{n=0}^{\rho_k-1}\|\mathcal{N}_k\|^nt^n$ a polynomial with nonnegative coefficients.

Let us consider the integral action
\begin{equation}
\hat{\mathcal{S}}_\kappa(t)
=\int_0^t\rmd s\,
\rme^{\kappa(t-s)\mathcal{L}_0}
\mathcal{D}_\kappa(s)
\rme^{\kappa s\mathcal{L}_0}.
\label{eqn:S12}
\end{equation}
Its peripheral part reads
\begin{equation}
\mathcal{P}_\varphi
\hat{\mathcal{S}}_\kappa(t)
\mathcal{P}_\varphi
=
\sum_{\alpha_k,\alpha_\ell\in\rmi\mathbb{R}}
\rme^{\kappa\alpha_kt}
\int_0^t\rmd s\,
\rme^{-\kappa(\alpha_k-\alpha_\ell)s}
\mathcal{P}_k
\mathcal{D}_\kappa(s)
\mathcal{P}_\ell.
\end{equation}
The other components of $\hat{\mathcal{S}}_\kappa(t)$ are bounded for $t\in[0,T]$ as follows:
\begingroup
\allowdisplaybreaks
\begin{align}
\|
\mathcal{Q}_\varphi
\hat{\mathcal{S}}_\kappa(t)
\mathcal{P}_\varphi
\|
&\le\int_0^t\rmd s\,
\|\rme^{\kappa(t-s)\mathcal{L}_0}
\mathcal{Q}_\varphi
\|
\|
\mathcal{D}_\kappa(s)
\|
\|
\rme^{\kappa s\mathcal{L}_0}
\mathcal{P}_\varphi
\|
\nonumber\\
&\le
D
\int_0^t\rmd s\,
\|
\rme^{\kappa(t-s)\mathcal{L}_0}
\mathcal{Q}_\varphi
\|
=
\frac{1}{\kappa}
D
\int_0^{\kappa t}\rmd s\,
\|
\rme^{s\mathcal{L}_0}
\mathcal{Q}_\varphi
\|
\le
\frac{DR}{\kappa\gamma_0}
,
\label{eqn:BoundRPQ}
\\
\|
\hat{\mathcal{S}}_\kappa(t)
\mathcal{Q}_\varphi
\|
&\le\int_0^t\rmd s\,
\|
\rme^{\kappa(t-s)\mathcal{L}_0}
\|
\|
\mathcal{D}_\kappa(s)
\|
\|
\rme^{\kappa s\mathcal{L}_0}
\mathcal{Q}_\varphi
\|
\nonumber\\
&\le
D
\int_0^t\rmd s\,
\|\rme^{\kappa s\mathcal{L}_0}
\mathcal{Q}_\varphi
\|
=
\frac{1}{\kappa}
D
\int_0^{\kappa t}\rmd s\,
\|\rme^{s\mathcal{L}_0}
\mathcal{Q}_\varphi
\|
\le
\frac{DR}{\kappa\gamma_0}.
\label{eqn:BoundRQ}
\end{align}
\endgroup
This observation leads us to the following theorem.
\begin{thm}
\label{thm:ConstL0}
Consider a  time-dependent generator $t\in[0,T]\mapsto\mathcal{L}_\kappa(t)$ on a finite-di\-men\-sion\-al Banach space of the form
\begin{equation}
\mathcal{L}_\kappa(t)
=\kappa\mathcal{L}_0+\mathcal{D}_\kappa(t),
\label{eqn:LkappaThm5}
\end{equation}
and let $t\mapsto\Lambda_\kappa(t)$ be the evolution generated by $\mathcal{L}_\kappa(t)$,
\begin{equation}
\Lambda_\kappa(t)
=\T\exp\!\left(
\int_0^t\rmd s\,\mathcal{L}_\kappa(s)
\right).
\end{equation}
Assume that $\mathcal{L}_0$ is the generator of a contraction semigroup, and let $\mathcal{P}_\varphi$ and $\mathcal{Q}_\varphi=1-\mathcal{P}_\varphi$ be the peripheral and nonperipheral projections of $\mathcal{L}_0$, respectively. 
Assume also that $t\mapsto \mathcal{D}_\kappa(t)$ is continuous and bounded by 
$\|\mathcal{D}_\kappa\|_{\infty,T}\le D$,
for all $\kappa>0$ with some $D\ge0$.

Now, if there exists a time-dependent generator $t\in[0,T]\mapsto\overline{\mathcal{D}}(t)=\mathcal{P}_\varphi\overline{\mathcal{D}}(t)\mathcal{P}_\varphi$ such that
\begin{equation}
\mathcal{S}_{\kappa,\varphi}(t)
=\int_0^t\rmd s\,
[
\rme^{\kappa(t-s)\mathcal{L}_0}
\mathcal{P}_\varphi
\mathcal{D}_\kappa(s)
\mathcal{P}_\varphi
\rme^{\kappa s\mathcal{L}_0}
-\rme^{\kappa t\mathcal{L}_0} \overline{\mathcal{D}}(s)
]
\to0,\quad\mathrm{as}\quad\kappa\to+\infty,
\label{eqn:CondThm}
\end{equation}
for $t\in[0,T]$, then one gets $\|\mathcal{S}_{\kappa,\varphi}\|_{\infty,T}\to0$, and
\begin{equation}
\Lambda_\kappa(t)
-
\rme^{\kappa t\mathcal{L}_0}
\overline{\Lambda}(t)
\to0,\qquad\mathrm{as}\qquad
\kappa\to+\infty,
\label{eqn:LimitThm}
\end{equation}
uniformly for $t\in[0,T]$, where $t\mapsto\overline{\Lambda}(t)$ is the evolution generated by $\overline{\mathcal{D}}(t)$,
\begin{equation}
\overline{\Lambda}(t)
=\T\exp\!\left(
\int_0^t\rmd s\,\overline{\mathcal{D}}(s)
\right).
\label{eqn:LambdaBarThm5}
\end{equation}
The convergence error is bounded by 
\begin{equation}
\|
\Lambda_\kappa(t)
-
\rme^{\kappa t\mathcal{L}_0}
\overline{\Lambda}(t)
\|
\le
\left(
\|\mathcal{S}_{\kappa,\varphi}\|_{\infty,T}
+\frac{2DR}{\kappa\gamma_0}
\right)
\rme^{\overline{D}T}
[
1
+
T\rme^{DT}
(D+\overline{D})],
\label{eqn:BoundThm2}
\end{equation}
where $\gamma_0>0$ is the smallest nonzero decay rate of $\mathcal{L}_0$ in~\eqref{eqn:eta}, 
$R\ge0$ is a constant bounding the integral of the decaying part of $\rme^{t\mathcal{L}_0}$ as $\|\rme^{t\mathcal{L}_0}\mathcal{Q}_\varphi\|_{1,\infty}\le R/\gamma_0$ in~\eqref{eqn:R}, and $\overline{D}=\|\overline{\mathcal{D}}\|_{\infty,T}$.
Moreover, 
\begin{equation}
\Lambda_\kappa(t)
-
\rme^{\kappa t\mathcal{L}_0}
\overline{\Lambda}(t)\mathcal{P}_\varphi
\to0,\qquad\mathrm{as}\qquad
\kappa\to+\infty,
\label{eqn:LimitThmP}
\end{equation}
uniformly on compact subsets of $(0,T]$, with a bound
\begin{align}
&
\|
\Lambda_\kappa(t)
-
\rme^{\kappa t\mathcal{L}_0}
\overline{\Lambda}(t)\mathcal{P}_\varphi
\|
\nonumber\\
&\qquad
\le
\left(
\|\mathcal{S}_{\kappa,\varphi}\|_{\infty,T}
+\frac{DR}{\kappa\gamma_0}
\right)
\rme^{\overline{D}T}
[
1
+
T\rme^{DT}
(D+\overline{D})
]
+
\frac{DR}{\kappa\gamma_0}
\rme^{DT}
+\rme^{-\kappa\gamma_0t}p(\kappa t),
\label{eqn:BoundThm2P}
\end{align}
where $p(t)$ is the polynomial in~\eqref{eq:polynil}.
\end{thm}
\begin{proof}
We use Lemma~\ref{lem:IntParts} for $\mathcal{L}_1(t)=\mathcal{L}_\kappa(t)=\kappa\mathcal{L}_0+\mathcal{D}_\kappa(t)$, $\mathcal{L}_2(t)=\kappa\mathcal{L}_0
+
\rme^{\kappa t\mathcal{L}_0}
\overline{\mathcal{D}}(t)
\rme^{-\kappa t\mathcal{L}_0}
$, and $\mathcal{L}_0(t)=\kappa\mathcal{L}_0$. 
These generate the evolutions $\Lambda_1(t)=\Lambda_\kappa(t)$, $\Lambda_2(t)=\rme^{\kappa t\mathcal{L}_0}\overline{\Lambda}(t)$, and $\Lambda_0(t)=\rme^{\kappa t\mathcal{L}_0}$, respectively. Accordingly, \eqref{eq:TildeLambda2} specializes to $\widetilde{\Lambda}_2(t)=\overline{\Lambda}(t)$ and $\widetilde{\mathcal{L}}_2(t)=\overline{\mathcal{D}}(t)$.
Define the integral action
\begin{align}
\mathcal{S}_{12}(t)
&=
\int_0^t\rmd s\,
\Lambda_0(t,s)
[
\mathcal{L}_1(s)
-
\mathcal{L}_2(s)
]
\Lambda_0(s)
\nonumber\\
&=
\int_0^t\rmd s\,
\rme^{\kappa t\mathcal{L}_0}
[
\rme^{-\kappa s\mathcal{L}_0}
\mathcal{D}_\kappa(s)
\rme^{\kappa s\mathcal{L}_0}
-
\overline{\mathcal{D}}(s)
]
=\mathcal{S}_\kappa(t).
\vphantom{\int_0^t}
\label{eqn:PeriActionThm}
\end{align}
It is bounded by
\begin{align}
\|
\mathcal{S}_\kappa(t)
\|
\le{}&
\|
\mathcal{P}_\varphi\mathcal{S}_\kappa(t)\mathcal{P}_\varphi
\|
+
\|
\mathcal{Q}_\varphi
\hat{\mathcal{S}}_\kappa(t)
\mathcal{P}_\varphi
\|
+
\|
\hat{\mathcal{S}}_\kappa(t)
\mathcal{Q}_\varphi
\|
\vphantom{\frac{2}{\kappa}}
\nonumber
\displaybreak[0]
\\
\le
{}&
\|
\mathcal{S}_{\kappa,\varphi}(t)
\|
+\frac{2DR}{\kappa\gamma_0},
\end{align}
where $\hat{\mathcal{S}}_\kappa(t)$ is defined in~\eqref{eqn:S12}, and we have used the inequalities~\eqref{eqn:BoundRPQ} and~\eqref{eqn:BoundRQ}.
Then, the bound~\eqref{eqn:BoundThm2} just follows from the bound~\eqref{eqn:KeyFormalismBound} of Lemma~\ref{lem:IntParts}\@. 
Recall that $\|\Lambda_{\kappa}(t,s)\|\le\rme^{DT}$, as shown in~\eqref{eqn:BoundedEvolution}, and $\|\overline{\Lambda}(t)\|\le\rme^{\overline{D}T}$. 
Note, in addition, that $\|\mathcal{D}_\kappa\|_{1,T}\le DT$ and $\|\overline{\mathcal{D}}\|_{1,T}\le\overline{D}T$.

The limit~\eqref{eqn:LimitThm} follows from the bound~\eqref{eqn:BoundThm2}, if we prove that the assumption~\eqref{eqn:CondThm} implies that $\|\mathcal{S}_{\kappa,\varphi}\|_{\infty, T}\to 0$ as $\kappa\to+\infty$. 
This is just a standard compactness argument:
Consider the family of functions $f_{\kappa}(t)=\|\mathcal{S}_{\kappa,\varphi}(t)\|$. From~\eqref{eqn:CondThm} and~\eqref{eq:PeripheralContraction}, we get that $\|\mathcal{S}_{\kappa,\varphi}(t)\| = \|\rme^{-\kappa t\mathcal{L}_0} \mathcal{S}_{\kappa,\varphi}(t)\|$, and 
\begin{align}
    |f_{\kappa}(t)-f_{\kappa}(u)| &\leq \|\rme^{-\kappa t\mathcal{L}_0}\mathcal{S}_{\kappa,\varphi}(t)- \rme^{-\kappa u\mathcal{L}_0}\mathcal{S}_{\kappa,\varphi}(u)\|
    \nonumber\\
    &\le
    \biggl|\int_u^t\rmd s\,\Bigl(
\|
\rme^{-\kappa s\mathcal{L}_0}
\mathcal{P}_\varphi\|
\|\mathcal{D}_\kappa(s)\|
\|\rme^{\kappa s\mathcal{L}_0} \mathcal{P}_\varphi\|
+ \|\overline{\mathcal{D}}(s)\|
\Bigr)
\biggr|
\nonumber \\
    &\le (D+\overline{D})|t-u|,
\end{align}
for $u,t\in[0,T]$, that is, the functions $f_\kappa$ are Lipschitz continuous  with a Lipschitz constant independent of $\kappa$. 
Since they converge pointwise to zero on the compact interval $[0,T]$, they converge uniformly, namely,
\begin{equation}
    \sup_{t\in[0,T]}|f_{\kappa}(t)|
    =\|\mathcal{S}_{\kappa,\varphi}\|_{\infty,T}\to 0, \qquad \mathrm{as} \qquad \kappa\to+\infty,
\end{equation}
and we are done.

Moreover, noting that $\overline{\Lambda}(t)\mathcal{P}_\varphi=\mathcal{P}_\varphi\overline{\Lambda}(t)\mathcal{P}_\varphi$ and $\overline{\mathcal{D}}(t)\mathcal{P}_\varphi=\mathcal{P}_\varphi\overline{\mathcal{D}}(t)$, and using the identity~\eqref{eqn:KeyFormalism} of Lemma~\ref{lem:IntParts}, we can bound the peripheral part as
\begin{align}
\|
[
\Lambda_\kappa(t)
-
\rme^{\kappa t\mathcal{L}_0}
\overline{\Lambda}(t)
]
\mathcal{P}_\varphi
\|
\le{}&
\|
\mathcal{S}_\kappa(t)
\mathcal{P}_\varphi
\|
\|\overline{\Lambda}(t){\mathcal{P}_\varphi}\|
\nonumber\\
&{}
+\int_0^t\rmd s\,\|\Lambda_\kappa(t,s)\|
\|
\mathcal{S}_\kappa(s)
\mathcal{P}_\varphi
\|
\|\overline{\Lambda}(s){\mathcal{P}_\varphi}\|
\,\Bigl(
\|\mathcal{D}_\kappa(s)\|
+
\|\overline{\mathcal{D}}(s)\|
\Bigr).
\label{eq:boundPeripheral0}
\end{align}
Now, using the fact that  $\mathcal{Q}_\varphi\overline{\mathcal{D}}(t)=0$ and the bound~\eqref{eqn:BoundRPQ}, we have
\begin{equation}
	\|\mathcal{S}_\kappa(t)\mathcal{P}_\varphi\|=\| \mathcal{S}_{\kappa,\varphi}(t)+\mathcal{Q}_\varphi\mathcal{S}_\kappa(t)\mathcal{P}_\varphi\|\leq \|\mathcal{S}_{\kappa,\varphi}(t) \|+\|\mathcal{Q}_\varphi \hat{\mathcal{S}}_\kappa(t)\mathcal{P}_\varphi\|\leq 
	\|
	\mathcal{S}_{\kappa,\varphi}
	\|_{\infty,T}
	+
	\frac{DR}{\kappa\gamma_0}.
\end{equation}
Using $\|\Lambda_\kappa(t,s)\|\le\rme^{DT}$, $\|\overline{\Lambda}(s)\mathcal{P}_\varphi\|\le\rme^{\overline{D}T}$, $\|\mathcal{D}_\kappa\|_{1,t}\le DT$, and $\|\overline{\mathcal{D}}\|_{1,t}\le\overline{D}T$ in~\eqref{eq:boundPeripheral0}, we get the bound on the peripheral evolution
\begin{equation}
	\|
	[
	\Lambda_\kappa(t)
	-
	\rme^{\kappa t\mathcal{L}_0}
	\overline{\Lambda}(t)
	]
	\mathcal{P}_\varphi
	\|
	\le{}
		\left(
		\|
		\mathcal{S}_{\kappa,\varphi}
		\|_{\infty,T}
		+
		\frac{DR}{\kappa\gamma_0}
		\right)
	\rme^{\overline{D}T}
	[
	1
	+
	T\rme^{DT}
	(D+\overline{D})].
	\label{eqn:BoundPeripheral}
\end{equation}
On the other hand, we have
\begin{equation}
\Lambda_\kappa(t)
=\rme^{\kappa t\mathcal{L}_0}
+\int_0^t\rmd s\,\Lambda_\kappa(t,s)\mathcal{D}_\kappa(s)\rme^{\kappa s\mathcal{L}_0},
\end{equation}
and its nonperipheral part is bounded as
\begin{align}
\|\Lambda_\kappa(t)\mathcal{Q}_\varphi\|
&\le\|\rme^{\kappa t\mathcal{L}_0}\mathcal{Q}_\varphi\|
+\int_0^t\rmd s\,\|\Lambda_\kappa(t,s)\|\|\mathcal{D}_\kappa(s)\|\|\rme^{\kappa s\mathcal{L}_0}\mathcal{Q}_\varphi\|
\nonumber\\
&\le\rme^{-\kappa\gamma_0t}p(\kappa t)
+\frac{DR}{\kappa\gamma_0}\rme^{DT},
\label{eqn:BoundNonPeripheral}
\end{align}
where we have used $\|\Lambda_\kappa(t,s)\|\le\rme^{DT}$, $\|\mathcal{D}_\kappa(s)\|\le D$, and the bounds~\eqref{eqn:NonperipheralDecay} and~\eqref{eqn:R}.
By combining the bounds~\eqref{eqn:BoundPeripheral} and~\eqref{eqn:BoundNonPeripheral}, and by using the triangle inequality $\|
\Lambda_\kappa(t)
-
\rme^{\kappa t\mathcal{L}_0}
\overline{\Lambda}(t)\mathcal{P}_\varphi
\|
\le
\|
[
\Lambda_\kappa(t)
-
\rme^{\kappa t\mathcal{L}_0}
\overline{\Lambda}(t)
]
\mathcal{P}_\varphi
\|
+
\|
\Lambda_\kappa(t)
\mathcal{Q}_\varphi
\|
$, we get the bound~\eqref{eqn:BoundThm2P}.
Under the condition~\eqref{eqn:CondThm}, the limit~\eqref{eqn:LimitThmP} follows.
\end{proof}

\begin{remark}[Effective generator]
Since
\begin{equation}
\rme^{\kappa t\mathcal{L}_0}
\overline{\Lambda}(t)
=\T\exp\!\left(
\int_0^t\rmd s\,[
\kappa \mathcal{L}_0
+
\rme^{\kappa s\mathcal{L}_0}
\overline{\mathcal{D}}(s)
\rme^{-\kappa s\mathcal{L}_0}
]
\right),
\end{equation}
Theorem~\ref{thm:ConstL0} states that, for large $\kappa$, the evolution $\Lambda_\kappa(t)$ is approximated by the evolution generated by the effective generator 
\begin{equation}\label{eq:effectiveGenerator}
\mathcal{L}_{\mathrm{eff},\kappa}(t)
=
\kappa\mathcal{L}_0
+
\rme^{\kappa t\mathcal{L}_0}
\overline{\mathcal{D}}(t)
\rme^{-\kappa t\mathcal{L}_0}.
\end{equation}
\end{remark}

\begin{remark}
There is no general recipe for constructing an effective generator $\overline{\mathcal{D}}(t)$ for generic $\mathcal{D}_\kappa(t)$ with general dependence on $t$ and $\kappa$, but for practically relevant cases, such as $\mathcal{D}_\kappa(t)=\mathcal{D}^{(1)}$ (constant), $\mathcal{D}^{(2)}(\kappa t)$ (fast driving), and $\mathcal{D}^{(3)}(t,\kappa t)$ (two driving timescales), their long-time averages on the peripheral subspace in the interaction picture, if they exist, provide effective generators as
\begin{align}
\overline{\mathcal{D}}^{(1)}
&=
\lim_{\tau\to+\infty}
\frac{1}{\tau}
\int_0^\tau\rmd s\,\rme^{-s\mathcal{L}_0}\mathcal{P}_\varphi\mathcal{D}^{(1)}\mathcal{P}_\varphi\rme^{s\mathcal{L}_0},
\\
\overline{\mathcal{D}}^{(2)}
&=
\lim_{\tau\to+\infty}
\frac{1}{\tau}
\int_0^\tau\rmd s\,\rme^{-s\mathcal{L}_0}\mathcal{P}_\varphi\mathcal{D}^{(2)}(s)\mathcal{P}_\varphi\rme^{s\mathcal{L}_0},
\label{eqn:AverageDt1}
\\
\overline{\mathcal{D}}^{(3)}(t)
&=
\lim_{\tau\to+\infty}
\frac{1}{\tau}
\int_0^\tau\rmd s\,\rme^{-s\mathcal{L}_0}\mathcal{P}_\varphi\mathcal{D}^{(3)}(t,s)\mathcal{P}_\varphi\rme^{s\mathcal{L}_0},
\label{eqn:AverageDt2}
\end{align}
respectively.
In the first two cases, the corresponding long-time average yields an effective generator satisfying condition~\eqref{eqn:CondThm}.
See Corollaries~\ref{cor:ConstL01}--\ref{cor:StrongLimit} and the examples studied in Sections~\ref{sec:Examples} and~\ref{sec:SecApp}\@.
In the two-timescale case, the pointwise existence of the  average~\eqref{eqn:AverageDt2} is not, in general, sufficient. 
Additional uniformity and regularity with respect to the frozen time $t$, such as the assumptions of Corollary~\ref{cor:ConstL02}, ensure that condition~\eqref{eqn:CondThm} is fulfilled.

In particular, for a constant generator $\mathcal{D}_\kappa(t)=\mathcal{D}^{(1)}$, the limit always exists and is given by the diagonal projection of $\mathcal{D}^{(1)}$ onto the peripheral eigenspaces of $\mathcal{L}_0$ in~\eqref{eqn:SpectralRep}: 
\begin{equation}
\overline{\mathcal{D}}^{(1)}
=\sum_{\alpha_k\in\rmi\mathbb{R}}\mathcal{P}_k \mathcal{D}^{(1)}\mathcal{P}_k.
\end{equation}
See Corollary~\ref{cor:StrongLimit}\@.
On the other hand, neither of the   limits~\eqref{eqn:AverageDt1} and~ \eqref{eqn:AverageDt2} exists in general.
A useful sufficient condition is given by a quasiperiodic dependence with finitely many frequencies $\nu\in\mathbb{R}$, that is,
\begin{equation}
   \mathcal{D}^{(2)}(s)= \sum_{\nu} \rme^{\rmi \nu s} \mathcal{D}^{(2)}_{\nu}, \qquad 
   \mathcal{D}^{(3)}(t,s)= \sum_{\nu} \rme^{\rmi \nu s} \mathcal{D}^{(3)}_{\nu}(t).
\end{equation}
If this is the case, the averages exist and retain precisely the blocks of the resonant components, namely,
\begin{equation}
   \overline{\mathcal{D}}^{(2)}= \sum_{\nu} \sum_{\alpha_k,\alpha_\ell \in \rmi\mathbb{R} } \delta_{\rmi \nu, \alpha_k - \alpha_\ell}   
   \mathcal{P}_k \mathcal{D}^{(2)}_{\nu} \mathcal{P}_{\ell}, \quad 
   \overline{\mathcal{D}}^{(3)}(t)= \sum_{\nu}\sum_{\alpha_k,\alpha_\ell \in \rmi\mathbb{R} } \delta_{\rmi \nu, \alpha_k - \alpha_\ell}   
   \mathcal{P}_k \mathcal{D}^{(3)}_{\nu}(t) \mathcal{P}_{\ell}.
\end{equation}
\end{remark}

If $\mathcal{L}_\kappa(t)$ of Theorem~\ref{thm:ConstL0} is the generator of a contraction semigroup, the bounds~\eqref{eqn:BoundThm2} and~\eqref{eqn:BoundThm2P} are free from the exponential growths $\rme^{DT}$ and $\rme^{\overline{D}T}$, as shown in the following theorem.
\begin{thm}[Improved bounds for generators of contraction semigroups]
\label{rmk:PhysicalBound}
In addition to the assumptions in Theorem~\ref{thm:ConstL0}, if $\mathcal{L}_\kappa(t)$ in~\eqref{eqn:LkappaThm5} is the generator of a contraction semigroup for each $t\in[0,T]$ and $\kappa>0$, then the bound~\eqref{eqn:BoundThm2} is simplified to 
\begin{equation}
\|
\Lambda_\kappa(t)
-
\rme^{\kappa t\mathcal{L}_0}
\overline{\Lambda}(t)
\|
\le
\|\overline{\Lambda}\|_{\infty,T}
\left(
\|\mathcal{S}_{\kappa,\varphi}\|_{\infty,T}
+\frac{2DR}{\kappa\gamma_0}
\right)
[
1
+
T
(D+\overline{D})],
\label{eqn:BoundThm6}
\end{equation}
for $t\in[0,T]$.

Furthermore, one has, instead of~\eqref{eqn:BoundThm2P},
\begin{equation}
\|
\Lambda_\kappa(t)
-
\rme^{\kappa t\mathcal{L}_0}
\overline{\Lambda}(t)\mathcal{P}_\varphi
\|
\le
\left(
\|\mathcal{S}_{\kappa,\varphi}\|_{\infty,T}
+\frac{DR}{\kappa\gamma_0}
\right)
[
1
+
T(D+\overline{D})]
+
\frac{DR}{\kappa\gamma_0}
+\rme^{-\kappa\gamma_0t}p(\kappa t),
\label{eqn:BoundThm2PPhys}
\end{equation}
for $t\in[0,T]$.
\end{thm}
\begin{proof}
Since $\mathcal{L}_\kappa(t)$ is the generator of a contraction semigroup for each $t\in[0,T]$ and $\kappa>0$, one has $\|\Lambda_\kappa(t,s)\|\le1$ for $0\le s\le t\le T$ and $\kappa>0$, as recalled in Proposition~\ref{prop:Contraction}\@.
Moreover, since $\mathcal{L}_0=\lim_{\kappa\to+\infty}\frac{1}{\kappa}\mathcal{L}_\kappa(t)$ is also the generator of a contraction semigroup, and hence $\|\rme^{t\mathcal{L}_0}\mathcal{P}_\varphi\|\le1$ for all $t\in\mathbb{R}$, the convergence~\eqref{eqn:LimitThmP} implies that
\begin{equation}
\|\overline{\Lambda}(t)\mathcal{P}_\varphi\|
=\lim_{\kappa\to+\infty}
\|
\rme^{-\kappa t\mathcal{L}_0}
\mathcal{P}_\varphi
\Lambda_\kappa(t)
\|
\le1,
\label{eqn:BarContraction}
\end{equation}
for $t\in[0,T]$.
That is, $\overline{\Lambda}(t)$ is guaranteed to be a contraction on the peripheral subspace $\mathcal{P}_\varphi$.
Then, using these in~\eqref{eq:boundPeripheral0}, 
the bound~\eqref{eqn:BoundThm2P} is simplified to~\eqref{eqn:BoundThm2PPhys}.

Finally, by using $\|\Lambda_\kappa(t,s)\|\leq 1$ in~\eqref{eqn:KeyFormalismBound}, together with
\begin{equation}
\|\widetilde{\Lambda}_2\|_{\infty,T}
=\|\overline{\Lambda}\|_{\infty,T},
\qquad
\|\mathcal{S}_\kappa\|_{\infty,T}
\le\|\mathcal{S}_{\kappa,\varphi}\|_{\infty,T}+\frac{2DR}{\kappa\gamma_0},
\end{equation}
and
\begin{equation}
\|\mathcal{D}_\kappa\|_{1,T}\le DT,
\qquad
\|\overline{\mathcal{D}}\|_{1,T}\le\overline D T,
\end{equation}
we obtain the bound~\eqref{eqn:BoundThm6}.
\end{proof}

Even if $\mathcal{L}_\kappa$ of Theorem~\ref{thm:ConstL0} is not the generator of a contraction semigroup, if its strong part $\mathcal{L}_0$ is the generator of a unitary evolution and if $\overline{\Lambda}(t,s)$ is uniformly bounded for $0\le s\le t\le T$, we have the following bound with a weaker exponential growth than $\rme^{DT}$.

\begin{thm}[Improved bound when the reference evolution is isometric]
\label{rmk:PhysicalBoundUnitary}
In addition to the assumptions in Theorem~\ref{thm:ConstL0}, if the reference generator $\mathcal{L}_0$ in~\eqref{eqn:LkappaThm5} is the generator of an isometric group and if $\overline{\Lambda}(t,s)$ constructed by~\eqref{eqn:LambdaBarThm5} is uniformly bounded, i.e.~$\|\overline{\Lambda}(t,s)\|\le \overline{L}$ for all $0\le s\le t\le T$, then one has, instead of~\eqref{eqn:BoundThm2},
\begin{equation}
\|
\Lambda_\kappa(t)
-
\rme^{\kappa t\mathcal{L}_0}
\overline{\Lambda}(t)
\|
\le \overline{L}\left[
\frac{1}{1-\|\mathcal{S}_\kappa\|_{\infty,T}}
\exp\!\left(
\frac{
\|\mathcal{S}_\kappa\|_{\infty,T}
}{1-\|\mathcal{S}_\kappa\|_{\infty,T}}
\overline{L}(D+\overline{D})T
\right)
-1\right],
\label{eqn:BoundThm2phys}
\end{equation}
for $t\in[0,T]$ and $\|\mathcal{S}_\kappa\|_{\infty,T}<1$, where $\mathcal{S}_\kappa(t)$ is the integral action defined in the second line of~\eqref{eqn:PeriActionThm}.
\end{thm}
\begin{proof}
To prove this bound, we need a strategy different from the one employed to get the bound~\eqref{eqn:KeyFormalismBound} in Lemma~\ref{lem:IntParts}\@.
The idea is to exploit the fact that, even though $\Lambda_\kappa(t)$ is not a contraction, if it is close to $\rme^{\kappa t\mathcal{L}_0}\overline{\Lambda}(t)$, which is bounded by $\overline{L}$, then the bound on $\|\Lambda_\kappa(t)\|$ is close to $\overline{L}$.

We apply the identity~\eqref{eqn:KeyFormalism} of Lemma~\ref{lem:IntParts} for 
$\Lambda_0(t)=\rme^{\kappa t\mathcal{L}_0}$, 
$\Lambda_1(t)=\rme^{\kappa t\mathcal{L}_0}\overline{\Lambda}(t)$, 
$\Lambda_2(t)=\Lambda_\kappa(t)$, $\widetilde{\Lambda}_2(t)=\rme^{-\kappa t\mathcal{L}_0}\Lambda_\kappa(t)$, 
$\mathcal{L}_1(t)-\mathcal{L}_0(t)=\rme^{\kappa t\mathcal{L}_0}\overline{\mathcal{D}}(t)\rme^{-\kappa t\mathcal{L}_0}$, 
$\widetilde{\mathcal{L}}_2(t)=\rme^{-\kappa t\mathcal{L}_0}\mathcal{D}_\kappa(t)\rme^{\kappa t\mathcal{L}_0}$, and 
\begin{align}
\mathcal{S}_{12}(t)
&=
\int_0^t\rmd s\,
\Lambda_0(t,s)
[
\mathcal{L}_1(s)
-
\mathcal{L}_2(s)
]
\Lambda_0(s)
\nonumber\\
&=
\int_0^t\rmd s\,
\rme^{\kappa(t-s)\mathcal{L}_0}
[
\rme^{\kappa s\mathcal{L}_0}\overline{\mathcal{D}}(s)\rme^{-\kappa s\mathcal{L}_0}
-
\mathcal{D}_\kappa(s)
]
\rme^{\kappa s\mathcal{L}_0}
=-\mathcal{S}_\kappa(t).
\vphantom{\int_0^t}
\end{align}
Noting that $\mathcal{L}_0$ is the generator of an isometric group and that $\overline{\Lambda}(t,s)$ is bounded by $\overline{L}$, we bound the identity
\begin{align}
\rme^{\kappa t\mathcal{L}_0}\overline{\Lambda}(t)
-
\Lambda_\kappa(t)
={}&{-\mathcal{S}_\kappa(t)}
\rme^{-\kappa t\mathcal{L}_0}\Lambda_\kappa(t)
\vphantom{\int_0^t}
\nonumber\\
&{}
-\int_0^t\rmd s\,
\rme^{\kappa t\mathcal{L}_0}\overline{\Lambda}(t,s)\rme^{-\kappa s\mathcal{L}_0}
\,\Bigl(
\rme^{\kappa s\mathcal{L}_0}\overline{\mathcal{D}}(s)\rme^{-\kappa s\mathcal{L}_0}
\mathcal{S}_\kappa(s)
\nonumber\\
&\hphantom{{}-\int_0^t\rmd s\,\rme^{\kappa t\mathcal{L}_0}\overline{\Lambda}(t,s)\rme^{-\kappa s\mathcal{L}_0}
\,\Bigl(}{}
-\mathcal{S}_\kappa(s)
\rme^{-\kappa s\mathcal{L}_0}\mathcal{D}_\kappa(s)\rme^{\kappa s\mathcal{L}_0}
\Bigr)\,
\rme^{-\kappa s\mathcal{L}_0}\Lambda_\kappa(s)
\end{align}
as
\begin{align}
\delta_\Lambda(t)
=
\|
\Lambda_\kappa(t)
-
\rme^{\kappa t\mathcal{L}_0}\overline{\Lambda}(t)
\|
\le{}&S
\|\Lambda_\kappa(t)\|
+
S\overline{L}M
\int_0^t\rmd s\,
\|\Lambda_\kappa(s)\|,
\end{align}
where $S=\|\mathcal{S}_\kappa\|_{\infty,T}$ and $M=D+\overline{D}$.
Notice here that
\begin{align}
\|\Lambda_\kappa(t)\|
\le{}&
\|\rme^{\kappa t\mathcal{L}_0}\overline{\Lambda}(t)\|
+
\|
\Lambda_\kappa(t)
-
\rme^{\kappa t\mathcal{L}_0}\overline{\Lambda}(t)
\|
\vphantom{\int_0^t}
\nonumber\\
\le{}&
\overline{L}
+
\delta_\Lambda(t)
\nonumber\\
\le{}&
\overline{L}
+
S\|\Lambda_\kappa(t)\|
+
S\overline{L}M
\int_0^t\rmd s\,
\|\Lambda_\kappa(s)\|
\nonumber\\
\le{}&
\overline{L}
+
S
[\overline{L}+\delta_\Lambda(t)]
+
S\overline{L}M
\int_0^t\rmd s\,
[\overline{L}+\delta_\Lambda(s)].
\label{eqn:QuasiContractive}
\end{align}
Now, the inequality~\eqref{eqn:QuasiContractive} yields
\begin{equation}
\overline{L}+\delta_\Lambda(t)
\le
\frac{\overline{L}}{1-S}
+
\frac{
S\overline{L}M
}{1-S}
\int_0^t\rmd s\,
[\overline{L}+\delta_\Lambda(s)].
\end{equation}
Then, by the Gr\"onwall inequality (see, e.g.~\cite[Theorem 1.10]{Tao2006Nonlinear}), 
we get
\begin{equation}
\overline{L}+\delta_\Lambda(t)
\le
\frac{\overline{L}}{1-S}
\exp\!\left(
\frac{
S
}{1-S}
\overline{L}Mt
\right),
\end{equation}
and the bound~\eqref{eqn:BoundThm2phys} follows.
\end{proof}

\begin{corol}[Rotating-wave approximation]\label{cor:ConstL01}
Assume that $\mathcal{D}_\kappa(t)$ in Theorem~\ref{thm:ConstL0} is of the form
\begin{equation}
\mathcal{D}_\kappa(t)
=\mathcal{D}(\kappa t),
\end{equation}
with $\mathcal{D}(t)$ continuous and bounded uniformly for all $t\ge0$, and that the following limit exists:
\begin{equation}
\lim_{\tau\to+\infty}
\frac{1}{\tau}
\int_0^\tau\rmd s\,\rme^{-s\mathcal{L}_0}\mathcal{P}_\varphi\mathcal{D}(s)\mathcal{P}_\varphi\rme^{s\mathcal{L}_0}
=\overline{\mathcal{D}}.
\label{eqn:AverageD}
\end{equation}
Then, one gets
\begin{equation}
\Lambda_\kappa(t)
-
\rme^{\kappa t\mathcal{L}_0}
\rme^{t\overline{\mathcal{D}}}
\to0,\qquad\mathrm{as}\qquad
\kappa\to+\infty,
\label{eqn:LimitCor3}
\end{equation}
uniformly for $t\in[0,T]$.
\end{corol}
\begin{proof}
We apply Theorem~\ref{thm:ConstL0} for $\mathcal{D}_\kappa(t)=\mathcal{D}(\kappa t)$ and $\overline{\mathcal{D}}(t)=\overline{\mathcal{D}}=\mathcal{P}_\varphi\overline{\mathcal{D}}\mathcal{P}_\varphi$.
In this case, we have
$D=\sup_{t\ge0}\|\mathcal{D}(t)\|$ and $\overline{D}=\|\overline{\mathcal{D}}\|\leq D$ [see~\eqref{eqn:AverageD} and~\eqref{eq:PeripheralContraction}],
and the condition~\eqref{eqn:CondThm} is fulfilled as
\begingroup
\allowdisplaybreaks
\begin{align}
\mathcal{S}_{\kappa,\varphi}(t)
&=
\int_0^t\rmd s\,
\rme^{\kappa t\mathcal{L}_0}
[
\rme^{-\kappa s\mathcal{L}_0}
\mathcal{P}_\varphi
\mathcal{D}(\kappa s)
\mathcal{P}_\varphi
\rme^{\kappa s\mathcal{L}_0}
-
\overline{\mathcal{D}}
]
\nonumber\\
&=
t
\rme^{\kappa t\mathcal{L}_0} \mathcal{P}_\varphi
\left(
\frac{1}{\kappa t}
\int_0^{\kappa t}\rmd s\,
[
\rme^{-s\mathcal{L}_0}
\mathcal{P}_\varphi
\mathcal{D}(s)
\mathcal{P}_\varphi
\rme^{s\mathcal{L}_0}
-
\overline{\mathcal{D}}
]
\right)
\nonumber\\
&\to0,\qquad\mathrm{as}\qquad\kappa\to+\infty,
\vphantom{\int_0^{\kappa t}\rmd s}
\end{align}
\endgroup
under the assumption~\eqref{eqn:AverageD}.
Therefore, the limit~\eqref{eqn:LimitCor3} holds by Theorem~\ref{thm:ConstL0}\@.
\end{proof}

\begin{corol}[Strong-coupling limit]\label{cor:StrongLimit}
    If $\mathcal{D}_\kappa(t)$ in Theorem~\ref{thm:ConstL0} is constant, i.e.~$\mathcal{D}_\kappa(t)=\mathcal{D}$, the long-time average~\eqref{eqn:AverageD} exists and is given by
    \begin{equation}\label{eq:DZ}
        \overline{\mathcal{D}}=\sum_{\alpha_k\in\rmi\mathbb{R}}\mathcal{P}_k\mathcal{D}\mathcal{P}_k
=\mathcal{D}_Z.
    \end{equation}
    Then, one has that
    \begin{equation}
        \rme^{t(\kappa\mathcal{L}_0+\mathcal{D})}
-
\rme^{t(\kappa\mathcal{L}_0+\mathcal{D}_Z)}\to 0,
\qquad\mathrm{as}\qquad
\kappa\to+\infty,
    \end{equation}
    with an error bound
    \begin{align}
\label{eq:StrongCouplingBound}
&\|
\rme^{t(\kappa\mathcal{L}_0+\mathcal{D})}
-
\rme^{t(\kappa\mathcal{L}_0+\mathcal{D}_Z)}
\|
\nonumber\\
&\qquad
\le
\frac{2}{\kappa}
\left(
\frac{m(m-1)}{\eta}P^2
+\frac{1}{\gamma_0}R
\right)\|\mathcal{D}\|
\rme^{t\|\mathcal{D}_Z\|}
\,\Bigl(
1
+
t\rme^{t\|\mathcal{D}\|}
(\|\mathcal{D}\|+\|\mathcal{D}_Z\|)
\Bigr),
    \end{align}
where $m$ is the number of distinct peripheral  eigenvalues $\{\alpha_k\}$ of $\mathcal{L}_0$, 
\begin{equation}
\eta=\mathop{\min_{\alpha_k,\alpha_\ell\in\rmi\mathbb{R}}}_{k\neq\ell}|\alpha_k-\alpha_\ell|
\end{equation}
is the minimal spectral gap in its peripheral spectrum (for $m=1$, set $\eta=\infty$), and $P=\max_k\|\mathcal{P}_k\|$ is the maximum norm of its peripheral spectral projections $\{\mathcal{P}_k\}$.
\end{corol}
The bound~\eqref{eq:StrongCouplingBound} is comparable with the ones presented in Refs.~\cite{burgarth2019generalized,gong2020error,burgarth2021eternal}.
\begin{proof}
The long-time average~\eqref{eqn:AverageD} reads
\begin{align}
\frac{1}{\tau}
\int_0^\tau\rmd s\,\rme^{-s\mathcal{L}_0}\mathcal{P}_\varphi\mathcal{D}\mathcal{P}_\varphi\rme^{s\mathcal{L}_0}
&=
\sum_{\alpha_k\in\rmi\mathbb{R}}\mathcal{P}_k\mathcal{D}\mathcal{P}_k
+
\mathop{\sum_{\alpha_k,\alpha_\ell\in\rmi\mathbb{R}}}_{k\neq\ell}
\frac{1-\rme^{-(\alpha_k-\alpha_\ell)\tau}}{(\alpha_k-\alpha_\ell)\tau}\mathcal{P}_k\mathcal{D}\mathcal{P}_\ell
\nonumber\\
&\to\sum_{\alpha_k\in\rmi\mathbb{R}}\mathcal{P}_k\mathcal{D}\mathcal{P}_k
=\mathcal{D}_Z,
\qquad\mathrm{as}\qquad
\tau\to+\infty,
\end{align}
where $\{\alpha_k\}$ and $\{\mathcal{P}_k\}$ are the spectrum and the spectral projections of $\mathcal{L}_0$, respectively, introduced in the spectral representation~\eqref{eqn:SpectralRep} of $\mathcal{L}_0$.
Therefore, Corollary~\ref{cor:ConstL01} applies.
The peripheral part of the integral action
\begin{equation}
\mathcal{S}_{\kappa,\varphi}(t)
=\int_0^t\rmd s\,
\rme^{\kappa t\mathcal{L}_0}
[\rme^{-\kappa s\mathcal{L}_0}\mathcal{P}_\varphi\mathcal{D}\mathcal{P}_\varphi\rme^{\kappa s\mathcal{L}_0}-\mathcal{D}_Z]
=
\mathop{\sum_{\alpha_k,\alpha_\ell\in\rmi\mathbb{R}}}_{k\neq\ell}
\frac{\rme^{\kappa\alpha_kt}-\rme^{\kappa\alpha_\ell t}}{\kappa(\alpha_k-\alpha_\ell)}\mathcal{P}_k\mathcal{D}\mathcal{P}_\ell
\end{equation}
is bounded by
\begin{equation}
\|\mathcal{S}_{\kappa,\varphi}(t)\|
\le
\mathop{\sum_{\alpha_k,\alpha_\ell\in\rmi\mathbb{R}}}_{k\neq\ell}
\left|
\frac{\sin(\kappa|\alpha_k-\alpha_\ell|t/2)}{\kappa|\alpha_k-\alpha_\ell|/2}
\right|
\|\mathcal{P}_k\mathcal{D}\mathcal{P}_\ell\|
\le
\frac{2m(m-1)}{\kappa\eta}P^2\|\mathcal{D}\|.
\label{eqn:SStrongCoupling}
\end{equation}
Thus, the bound~\eqref{eqn:BoundThm2} presented in Theorem~\ref{thm:ConstL0} yields~\eqref{eq:StrongCouplingBound}.
\end{proof}

\begin{corol}[Rotating-wave approximation with two driving timescales]\label{cor:ConstL02}
If $\mathcal{D}_\kappa(t)$ in Theorem~\ref{thm:ConstL0} is of the form
\begin{equation}
\mathcal{D}_\kappa(t)
=\mathcal{D}(t,\kappa t),
\end{equation}
with $\mathcal{D}(t,s)$ being continuously differentiable and uniformly bounded  for all $t\in[0,T]$ and $s\ge0$, and if the following limits exist:
\begin{equation}
\overline{\mathcal{D}}(t,\tau)
=
\frac{1}{\tau}
\int_0^\tau\rmd s\,\rme^{-s\mathcal{L}_0}\mathcal{P}_\varphi\mathcal{D}(t,s)\mathcal{P}_\varphi\rme^{s\mathcal{L}_0}
\to\overline{\mathcal{D}}(t),
\quad
\partial_t\overline{\mathcal{D}}(t,\tau)
\to\partial_t\overline{\mathcal{D}}(t),
\quad\mathrm{as}\quad\tau\to+\infty,
\label{eqn:AverageDt}
\end{equation}
uniformly for $t\in[0,T]$, then one gets
\begin{equation}
\Lambda_\kappa(t)
-
\rme^{\kappa t\mathcal{L}_0}
\T\exp\!\left(
\int_0^t\rmd s\,\overline{\mathcal{D}}(s)
\right)
\to0,\qquad\mathrm{as}\qquad
\kappa\to+\infty,
\label{eqn:LimitCor4}
\end{equation}
uniformly for $t\in[0,T]$.
\end{corol}
\begin{proof}
We apply Theorem~\ref{thm:ConstL0} for $\mathcal{D}_\kappa(t)=\mathcal{D}(t,\kappa t)$ and $\overline{\mathcal{D}}(t)=\mathcal{P}_\varphi\overline{\mathcal{D}}(t)\mathcal{P}_\varphi$.
In this case, we have
$D=\sup_{t\in[0,T],s\ge0}\|\mathcal{D}(t,s)\|$ and $\overline{D}=\sup_{t\in[0,T]}\|\overline{\mathcal{D}}(t)\|\leq D$ [see~\eqref{eqn:AverageDt} and~\eqref{eq:PeripheralContraction}].
Define
\begin{equation}
\hat{\mathcal{D}}(t,s)
=\rme^{-s\mathcal{L}_0}\mathcal{P}_\varphi
\mathcal{D}(t,s)
\mathcal{P}_\varphi\rme^{s\mathcal{L}_0},
\end{equation}
and note that
\begin{equation}
\frac{\rmd}{\rmd s}
\int_0^s\rmd u\,\hat{\mathcal{D}}(s,\kappa u)
=\hat{\mathcal{D}}(s,\kappa s)
+
\int_0^s\rmd u\,\partial_s{\hat{\mathcal{D}}}(s,\kappa u).
\end{equation}
Then, we have
\begin{align}
\int_0^t\rmd s\,\hat{\mathcal{D}}(s,\kappa s)
&=
\int_0^t\rmd u\,\hat{\mathcal{D}}(t,\kappa u)
-
\int_0^t\rmd s\int_0^s\rmd u\,\partial_s\hat{\mathcal{D}}(s,\kappa u)
\nonumber\\
&=
t\overline{\mathcal{D}}(t,\kappa t)
-
\int_0^t\rmd s\,s\partial_1\overline{\mathcal{D}}(s,\kappa s),
\end{align}
where $\partial_1\overline{\mathcal{D}}(s_1,s_2)$ denotes the derivative with respect to the first argument $s_1$ of $\overline{\mathcal{D}}(s_1,s_2)$, and the condition~\eqref{eqn:CondThm} is fulfilled as
\begin{align}
\mathcal{S}_{\kappa,\varphi}(t)
&=
\int_0^t\rmd s\,
\rme^{\kappa t\mathcal{L}_0}
[
\rme^{-\kappa s\mathcal{L}_0}
\mathcal{P}_\varphi
\mathcal{D}(s,\kappa s)
\mathcal{P}_\varphi
\rme^{\kappa s\mathcal{L}_0}
-
\overline{\mathcal{D}}(s)
]
\nonumber\\
&=
\rme^{\kappa t\mathcal{L}_0}
\int_0^t\rmd s\,
[
\hat{\mathcal{D}}(s,\kappa s)
-
\overline{\mathcal{D}}(s)
]
\nonumber\\
&=
\rme^{\kappa t\mathcal{L}_0} \mathcal{P}_\varphi
\left(
t[
\overline{\mathcal{D}}(t,\kappa t)
-
\overline{\mathcal{D}}(t)
]
-
\int_0^t\rmd s\,s
[
\partial_1\overline{\mathcal{D}}(s,\kappa s)
-
\partial_s\overline{\mathcal{D}}(s)
]
\right)
\nonumber\\
&\to0,\qquad\mathrm{as}\qquad\kappa\to+\infty,
\vphantom{\int_0^{\kappa t}\rmd s}
\end{align}
under the assumptions in~\eqref{eqn:AverageDt}.
Therefore, the limit~\eqref{eqn:LimitCor4} holds by Theorem~\ref{thm:ConstL0}\@.
\end{proof}

\section{Examples}
\label{sec:Examples}
In this section, we show how the general results specialize to some examples commonly encountered in several settings. In the first two examples, we consider the paradigmatic situation of the RWA applied to a qubit system, and we consider two different dissipators for the additional noise term: in one case (Example~\ref{ex:2}), the dissipator commutes with the rotating reference frame, and as a consequence the dissipator remains unchanged in the approximation; in the other case (Example~\ref{ex:3}), the dissipator is modified in the approximation. Finally, Example~\ref{ex:4} serves as an illustration of a setting where the rotating reference frame is also ``shrinking'', i.e.~the generator of the reference frame evolution $\mathcal{L}_0$ contains dissipative terms.
\begin{example}\label{ex:2} 
Let us consider the CPTP evolution $\Lambda_\omega(t)$ generated by the time-dependent GKLS generator
\begin{equation}
\mathcal{L}_\omega(t) \varrho
=
-\rmi
\left[
\frac{1}{2}\omega Z+g\cos(\omega t)X,\varrho
\right]
-\frac{1}{2}\gamma(\varrho-Z\varrho Z),
\label{eqn:Example2}
\end{equation}
where $X$, $Y$, and $Z$ are the first, second, and third Pauli operators, respectively, and $\varrho$ is a 2$\times$2 matrix.
This describes the evolution of a qubit driven under dephasing noise, with a dephasing rate $\gamma\ge0$. The generator takes the form  $\mathcal{L}_\omega(t)=\omega\mathcal{L}_0+\mathcal{D}_\omega(t)$, with
\begin{equation}
    \omega\mathcal{L}_0\varrho=-\frac{\rmi}{2}\omega[Z,\varrho], \qquad \mathcal{D}_\omega(t)\varrho=-\rmi g\cos(\omega t)
[X,\varrho]
-\frac{1}{2}\gamma(\varrho-Z\varrho Z). 
\label{eq:Dex1}
\end{equation}
In this case, the spectrum of $\mathcal{L}_0$ consists only of purely imaginary eigenvalues ($\mathcal{P}_\varphi=1$ and $R/\gamma_0=0$), and $\rme^{\omega t\mathcal{L}_0}$ is unitary,
\begin{equation}\label{eq:refEx1}
    \rme^{\omega t\mathcal{L}_0}\varrho=\rme^{-\frac{\rmi}{2}\omega tZ}\varrho \rme^{\frac{\rmi}{2}\omega tZ}.
\end{equation}
Then, using Proposition~\ref{prop:UConj} in Appendix~\ref{app:ElementaryFacts} and taking into account the fact that the dissipation in~\eqref{eq:Dex1} is unchanged by the action of~\eqref{eq:refEx1}, we get the generator in the rotating reference frame 
\begin{equation}
\hat{\mathcal{D}}(\omega t)\varrho
=\rme^{-\omega t\mathcal{L}_0}\mathcal{D}_\omega(t)\rme^{\omega t\mathcal{L}_0} \varrho
=
-\frac{\rmi}{2}g
\,\Bigl[
[1+\cos(2\omega t)]X-\sin(2\omega t)Y,\varrho
\Bigr]
-\frac{1}{2}\gamma(\varrho-Z\varrho Z).
\end{equation}
Its long-time average converges as
\begin{align}
\frac{1}{\tau}\int_0^\tau\rmd s\,\hat{\mathcal{D}}(s) \varrho
&=
-\frac{\rmi}{2}g
\left[
\left(1+\frac{\sin(2\tau)}{2\tau}\right)X-\frac{1-\cos(2\tau)}{2\tau}Y
,\varrho
\right]
-\frac{1}{2}\gamma(\varrho-Z\varrho Z)
\nonumber\\
&\to
-\frac{\rmi}{2}g
[
X
,\varrho
]
-\frac{1}{2}\gamma(\varrho-Z\varrho Z)
=\overline{\mathcal{D}}\varrho,
\vphantom{\frac{\sin2\Omega \tau}{2\Omega\tau}}
\label{eqn:DbarEx1}
\end{align}
in the limit $\tau\to+\infty$.
Corollary~\ref{cor:ConstL01} applies, and one gets
\begin{equation}
\Lambda_\omega(t)-\rme^{\omega t\mathcal{L}_0}\rme^{t\overline{\mathcal{D}}}\to0,\qquad\mathrm{as}\qquad\omega\to+\infty.
\end{equation}
This means that the evolution $\Lambda_\omega(t)$ generated by $\mathcal{L}_\omega(t)$ is approximated by the evolution $\Lambda_\mathrm{RWA}(t)$ generated by the effective generator
\begin{align}
\mathcal{L}_\mathrm{RWA}(t)\varrho
&=\omega\mathcal{L}_0\varrho+\rme^{\omega t\mathcal{L}_0}\overline{\mathcal{D}}\rme^{-\omega t\mathcal{L}_0}\varrho
\nonumber\\
&=
-\rmi
\left[
\frac{1}{2}\omega Z
+
\frac{1}{2}g[\cos(\omega t)X+\sin(\omega t)Y],\varrho
\right]
-\frac{1}{2}\gamma(\varrho-Z\varrho Z).
\end{align}
This is an example of the RWA\@.
The relevant action is given by
\begin{equation}
\mathcal{S}_\omega(t) \varrho
=\rme^{\omega t\mathcal{L}_0}
\int_0^t\rmd s\,[\hat{\mathcal{D}}(\omega s)-\overline{\mathcal{D}}]\varrho
=
-\frac{\rmi g}{2\omega}
\sin(\omega t)[
X,\rme^{-\frac{\rmi}{2}\omega tZ}\varrho\,\rme^{\frac{\rmi}{2}\omega tZ}
],
\label{eqn:SEx1}
\end{equation}
and we have
\begingroup
\allowdisplaybreaks
\begin{gather}
\|\mathcal{S}_\omega(t)\|_\diamond
=\frac{|g \sin(\omega t)|}{\omega},
\label{eqn:DiamondEx1a}
\\
\|\mathcal{D}_\omega(t)\|_\diamond
=\frac{1}{2}\,\Bigl(\gamma+\sqrt{\gamma^2+16g^2\cos^2(\omega t)}\Bigr),\qquad
\|\overline{\mathcal{D}}\|_\diamond
=\frac{1}{2}\,\Bigl(\gamma+\sqrt{\gamma^2+4g^2}\Bigr),
\label{eqn:DiamondEx1b}
\end{gather}
\endgroup
where we have used the formula~\eqref{eqn:DiamondPhi1} of Example~\ref{ex:QubitMaps} in Appendix~\ref{app:DiamondNorm} to compute these diamond norms.
Note that the diamond norm is unitarily invariant and the unitary $\rme^{-\frac{\rmi}{2}\omega t\mathcal{Z}}$, with $\mathcal{Z}\varrho=[Z,\varrho]$, can be removed from $\mathcal{S}_\omega(t)$ to compute its diamond norm by $\|\mathcal{S}_\omega(t)\|_\diamond=\|\mathcal{S}_\omega(t)\rme^{\frac{\rmi}{2}\omega t\mathcal{Z}}\|_\diamond$.
Since the generator $\mathcal{L}_\omega(t)$ in~\eqref{eqn:Example2} is a physically valid GKLS generator, the evolution $\Lambda_\omega(t)$ it generates is CPTP and its diamond norm is $\|\Lambda_\omega(t)\|_\diamond=1$.
The effective generator $\overline{\mathcal{D}}$ in~\eqref{eqn:DbarEx1} in the RWA is also of the GKLS form, and $\|\overline{\Lambda}(t)\|_\diamond=1$\@.
Then, Theorem~\ref{rmk:PhysicalBound} applies, and a bound on the error of the RWA is provided by~\eqref{eqn:BoundThm6} as
\begin{equation}
\|
\Lambda_\omega(t)
-
\Lambda_\mathrm{RWA}(t)
\|_\diamond
\le
\frac{|g|}{\omega}
[
1
+
(2\gamma+3|g|)t
],
\label{eq:boundEx1}
\end{equation}
where we have used the inequality $\sqrt{x^2+y^2}\le|x|+|y|$ for real numbers $x$ and $y$ to simplify the bound.
In particular, the uniform bound
\begin{equation}
\sup_{t\in[0,T]}\|
\Lambda_\omega(t)
-
\Lambda_\mathrm{RWA}(t)
\|_\diamond
\le
\frac{|g|}{\omega}
[
1
+
(2\gamma+3|g|)T
]
\label{eq:boundEx1uniform}
\end{equation}
holds for compact intervals $[0,T]$\@. 
See Fig.~\ref{fig:Examples1-2}\@.
\end{example}
\begin{figure}
    \centering
    \includegraphics[width=0.49\linewidth]{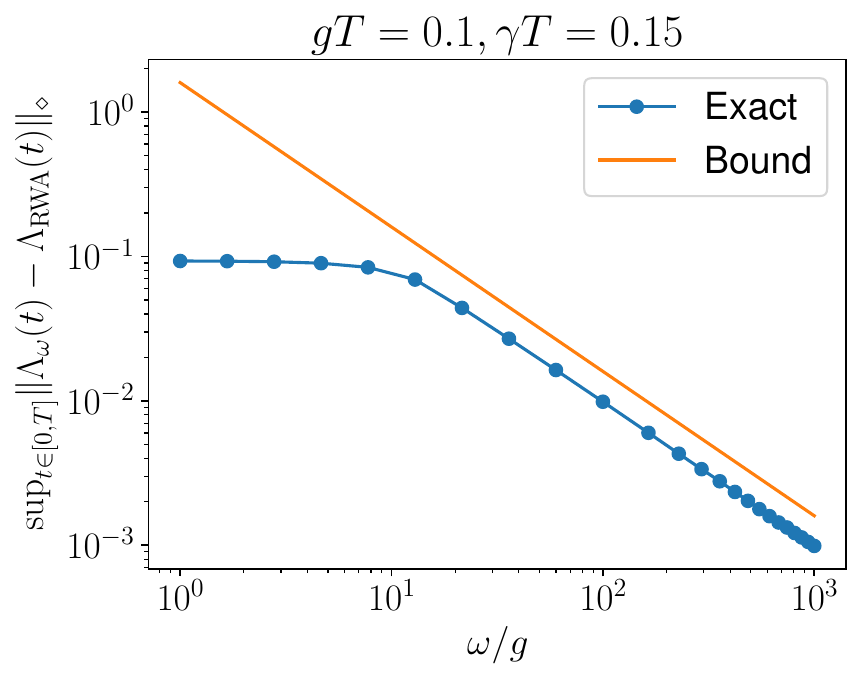}~
    \includegraphics[width=0.49\linewidth]{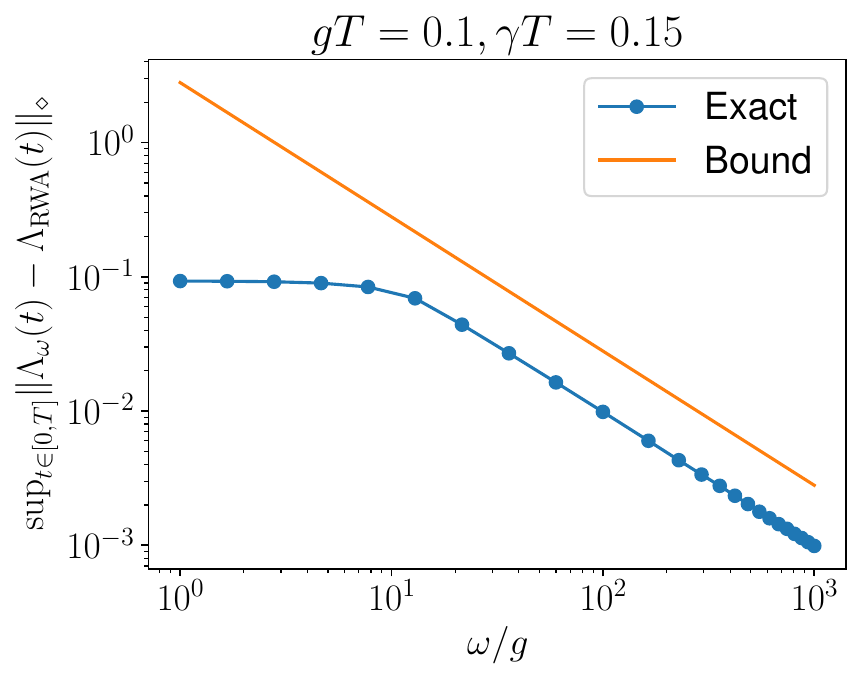}
    \caption{Comparison of the exact diamond distances computed numerically and  the bounds~\eqref{eq:boundEx1uniform} and~\eqref{eq:boundEx2uniform} obtained for Examples~\ref{ex:2} (Left) and~\ref{ex:3} (Right).
    The onset of the fast-driving regime is marked by the knee at $\omega/g\simeq 10$, beyond which oscillatory averaging becomes effective and the error approaches the asymptotic $1/\omega$ scaling.}
    \label{fig:Examples1-2}
\end{figure}

\begin{example}\label{ex:3}
Let us consider the evolution $\Lambda_\omega(t)$ generated by the time-dependent GKLS generator
\begin{equation}
\mathcal{L}_\omega(t) \varrho
=
-\rmi
\left[
\frac{1}{2}\omega Z+g\cos(\omega t)X,\varrho
\right]
-\frac{1}{2}\gamma(\varrho-X\varrho X).
\label{eqn:Example3}
\end{equation}
In this case, the nonunitary part will also be modified in the high-frequency limit.
Let us split this generator as $\mathcal{L}_\omega(t)=\omega\mathcal{L}_0+\mathcal{D}_\omega(t)$, with $\omega\mathcal{L}_0\varrho=-\frac{\rmi}{2}\omega[Z,\varrho]$ and $\mathcal{D}_\omega(t) \varrho=-\rmi g\cos(\omega t)[X,\varrho]-\frac{1}{2}\gamma(\varrho-X\varrho X)$.
The generator in the reference frame rotating with $\rme^{\omega t\mathcal{L}_0}$ reads
\begingroup
\allowdisplaybreaks
\begin{align}
\hat{\mathcal{D}}(\omega t)\varrho
={}&
\rme^{-\omega t\mathcal{L}_0}\mathcal{D}_\omega(t)\rme^{\omega t\mathcal{L}_0}\varrho
\nonumber\\
={}&
{-\frac{\rmi}{2}g}
\,\Bigl[
[1+\cos(2\omega t)]X-\sin(2\omega t)Y,{}\varrho{}
\Bigr]
\nonumber\\
&{}
-\frac{1}{4}\gamma\,\Bigl(
2{}\varrho{}
-[1+\cos(2\omega t)]
X
{}\varrho{}
X
-[1-\cos(2\omega t)]
Y
{}\varrho{}
Y
\nonumber\\
&\hphantom{{}-\frac{1}{4}\gamma\,\Bigl(2{}\varrho{}-[1+\cos(2\omega t)]X{}\varrho{}X}
{}+\sin(2\omega t)
(
X
{}\varrho{}
Y
+
Y
{}\varrho{}
X
)
\Bigr).
\end{align}
\endgroup
Its long-time average converges as
\begingroup
\allowdisplaybreaks
\begin{align}
\frac{1}{\tau}\int_0^\tau\rmd s\,\hat{\mathcal{D}}(s)\varrho
={}&
{-\frac{\rmi}{2}g}
\left[
\left(1+\frac{\sin(2\tau)}{2\tau}\right)X-\frac{1-\cos(2\tau)}{2\tau}Y
,{}\varrho{}
\right]
\nonumber\\
&{}
-\frac{1}{4}\gamma\,\biggl[
2{}\varrho{}
-
\left(
1+\frac{\sin(2\tau)}{2\tau}
\right)
X
{}\varrho{}
X
-
\left(
1-\frac{\sin(2\tau)}{2\tau}
\right)
Y
{}\varrho{}
Y
\nonumber\\
&\qquad\qquad\qquad\qquad\qquad\qquad\quad\ \ {}
{}
+
\frac{1-\cos(2\tau)}{2\tau}
(
X
{}\varrho{}
Y
+
Y
{}\varrho{}
X
)
\biggr]
\nonumber\\
\to{}&
{-\frac{\rmi}{2}g}
[
X
,{}\varrho{}
]
-\frac{1}{4}\gamma(2{}\varrho{}-X{}\varrho{}X-Y{}\varrho{}Y)
=\overline{\mathcal{D}}\varrho,
\vphantom{\frac{\sin2\Omega\tau}{2{\Omega}\tau}}
\label{eqn:DbarEx2}
\end{align}
\endgroup
in the limit $\tau\to+\infty$.
Corollary~\ref{cor:ConstL01} applies, and one gets
\begin{equation}
\Lambda_\omega(t)-\rme^{\omega t\mathcal{L}_0}\rme^{t\overline{\mathcal{D}}}\to0,\qquad\mathrm{as}\qquad\omega\to+\infty.
\end{equation}
This means that the evolution $\Lambda_\omega(t)$ generated by $\mathcal{L}_\omega(t)$ is approximated by the evolution $\Lambda_\mathrm{RWA}(t)$ generated by the effective generator
\begin{align}
\mathcal{L}_\mathrm{RWA}(t)\varrho
={}&\omega\mathcal{L}_0\varrho+\rme^{\omega t\mathcal{L}_0}\overline{\mathcal{D}}\rme^{-\omega t\mathcal{L}_0}\varrho
\nonumber\\
={}&
{-\rmi}
\left[
\frac{1}{2}\omega Z
+
\frac{1}{2}g[\cos(\omega t)X+\sin(\omega t)Y],{}\varrho{}
\right]
-\frac{1}{4}\gamma
(
2{}\varrho{}
-X{}\varrho{}X
-Y{}\varrho{}Y
).
\end{align}
The relevant action is given by
\begin{align}
\mathcal{S}_\omega(t)\varrho
={}&\rme^{\omega t\mathcal{L}_0}
\int_0^t\rmd s\,[\hat{\mathcal{D}}(\omega s)-\overline{\mathcal{D}}]\varrho
\nonumber\\
={}&
{-\frac{\rmi g}{2\omega}}\sin(\omega t)
[
X,
\rme^{-\frac{\rmi}{2}\omega tZ}{}\varrho{}\rme^{\frac{\rmi}{2}\omega tZ}
]
+\frac{\gamma}{4\omega}\sin(\omega t)
(
X{}\varrho{}X
-
Y{}\varrho{}Y
),
\label{eqn:Ex3S}
\end{align}
and we have
\begin{gather}
\|\mathcal{S}_\omega(t)\|_\diamond
\le
\frac{|{\sin(\omega t)}|}{\omega}
\left(
|g|+\frac{1}{2}\gamma
\right),
\label{eqn:DiamondEx2a}
\\
\|\mathcal{D}_\omega(t)\|_\diamond
=\sqrt{\gamma^2+4g^2\cos^2(\omega t)}
,\qquad
\|\overline{\mathcal{D}}\|_\diamond
=\frac{1}{4}\,\Bigl(\gamma+\sqrt{9\gamma^2+16g^2}\Bigr).
\label{eqn:DiamondEx2b}
\end{gather}
Here, for $\|\mathcal{D}_\omega(t)\|_\diamond$ and $\|\overline{\mathcal{D}}\|_\diamond$, we have used the formulas~\eqref{eqn:DiamondPhi2} and~\eqref{eqn:DiamondPhi3} of Example~\ref{ex:QubitMaps} in Appendix~\ref{app:DiamondNorm}\@.
For $\|\mathcal{S}_\omega(t)\|_\diamond$, we have computed the diamond norms of the two terms in~\eqref{eqn:Ex3S} separately using the formulas~\eqref{eqn:DiamondPhi2} and~\eqref{eqn:DiamondPhi4} of Example~\ref{ex:QubitMaps} in Appendix~\ref{app:DiamondNorm}, and have bounded $\|\mathcal{S}_\omega(t)\|_\diamond$ by the triangle inequality.
Since the generator $\mathcal{L}_\omega(t)$ in~\eqref{eqn:Example3} is a physically valid GKLS generator, the evolution $\Lambda_\omega(t)$ it generates is CPTP and its diamond norm is $\|\Lambda_\omega(t)\|_\diamond=1$.
The effective generator $\overline{\mathcal{D}}$ in~\eqref{eqn:DbarEx2} in the RWA is also of the GKLS form, and $\|\overline{\Lambda}(t)\|_\diamond=1$\@.
Then, Theorem~\ref{rmk:PhysicalBound} applies, and a bound on the error of the RWA is provided by~\eqref{eqn:BoundThm6} as
\begin{equation}
\|
\Lambda_\omega(t)
-
\Lambda_\mathrm{RWA}(t)
\|_\diamond
\le
\frac{1}{\omega}\left(
|g|+\frac{1}{2}\gamma
\right)
[
1
+
(
2\gamma+3|g|
)t
],
\label{eq:boundEx2}
\end{equation}
where we have again used the inequality $\sqrt{x^2+y^2}\le|x|+|y|$ for real numbers $x$ and $y$.
In particular, the uniform bound
\begin{equation}
\sup_{t\in[0,T]}
\|
\Lambda_\omega(t)
-
\Lambda_\mathrm{RWA}(t)
\|_\diamond
\le
\frac{1}{\omega}\left(
|g|+\frac{1}{2}\gamma
\right)
[
1
+
(
2\gamma+3|g|
)T]
\label{eq:boundEx2uniform}
\end{equation}
holds for compact intervals $[0,T]$\@. 
See Fig.~\ref{fig:Examples1-2}\@.
\end{example}

\begin{example}\label{ex:4}
Let us consider a three-level system with three levels $\{\ket{0},\ket{1},\ket{2}\}$, whose evolution $\Lambda_{\omega,\kappa}(t)$ is generated by the time-dependent GKLS generator
\begin{equation}
\mathcal{L}_{\omega,\kappa}(t)
=\mathcal{L}_{0,\omega,\kappa}+\mathcal{D}_\omega(t),
\label{eqn:Example4a}
\end{equation}
with
\begingroup
\allowdisplaybreaks
\begin{gather}\label{eq:L0Ex3}
\mathcal{L}_{0,\omega,\kappa}\varrho
=-\rmi\omega[H,{}\varrho{}]
-\kappa\,\Bigl(
\ket{2}\bra{2}{}\varrho{}
+
{}\varrho{}\ket{2}\bra{2}
-
2\ket{2}\bra{2}{}\varrho{}\ket{2}\bra{2}
\Bigr),
\quad
H=\ket{1}\bra{1}+2\ket{2}\bra{2}
,\\
\mathcal{D}_\omega(t)
=-\rmi \cos(\omega t) (g_1\mathcal{X}_{01}+g_2\mathcal{X}_{12}),
\label{eqn:Example4b}
\end{gather}
\endgroup
where  $\mathcal{X}_{ij}\varrho=[X_{ij},{}\varrho{}]$ with $X_{ij}=\ket{i}\bra{j}+\ket{j}\bra{i}$ for $i,j=0,1,2$\@. 
This describes a three-level system driven under strong dephasing between the two sectors $\{\ket{0},\ket{1}\}$ and $\{\ket{2}\}$.
We analyze the evolution $\Lambda_{\omega,\kappa}(t)$ in the limit $\omega,\kappa\to+\infty$, with the ratio $\omega/\kappa$ fixed constant.
This limit is closely related to the strong-decoherence and ultrastrong system-bath coupling regimes studied in Refs.~\cite{Trushechkin2022,MarcantoniMerkli2025}, where rapid decoherence suppresses coherences between the sectors selected by the system-bath interaction and gives rise to an effective projected Zeno dynamics.

In this case, the strong part $\mathcal{L}_{0,\omega,\kappa}$ of the generator generates a nonunitary evolution
\begin{equation}\label{eq:shrinkingFrame}
\rme^{t\mathcal{L}_{0,\omega,\kappa}}
=\rme^{-\rmi\omega t\mathcal{H}}
(
\mathcal{P}_\varphi+\rme^{-\kappa t}\mathcal{Q}_\varphi
),
\end{equation}
with its peripheral and nonperipheral projections given by
\begin{equation}
\begin{cases}
\medskip
\displaystyle
\mathcal{P}_\varphi\varrho=P_1{}\varrho{}P_1+P_2{}\varrho{}P_2,\\
\displaystyle
\mathcal{Q}_\varphi\varrho=P_1{}\varrho{}P_2+P_2{}\varrho{}P_1,
\end{cases}
\qquad
P_1=\ket{0}\bra{0}+\ket{1}\bra{1},
\quad
P_2=\ket{2}\bra{2},
\end{equation}
and the unitary generator $-\rmi\mathcal{H}\varrho=-\rmi[H,{}\varrho{}]$. 
Since $\|\mathcal{Q}_\varphi\|_\diamond=1$ (see Example~\ref{ex:NormQ} in Appendix~\ref{app:DiamondNorm}), we have bounds
\begin{equation}
\|\rme^{t\mathcal{L}_{0,\omega,\kappa}}\mathcal{Q}_\varphi\|_\diamond=\rme^{-\kappa t},\qquad
\int_0^\infty\rmd s\,\|\rme^{s\mathcal{L}_{0,\omega,\kappa}}\mathcal{Q}_\varphi\|_\diamond
=\frac{1}{\kappa},
\end{equation}
on the nonperipheral part of the evolution $\rme^{t\mathcal{L}_{0,\omega,\kappa}}$.

Let us look at the generator in the reference frame evolving with $\rme^{t\mathcal{L}_{0,\omega,\kappa}}$. Using~\eqref{eq:shrinkingFrame} and Proposition~\ref{prop:blockdiagsuper}, one gets
\begin{align}
\rme^{-t\mathcal{L}_{0,\omega,\kappa}}
\mathcal{D}_\omega(t)
\rme^{t\mathcal{L}_{0,\omega,\kappa}}
={}&{-\frac{\rmi}{2}}
g_1
\,\Bigl(
[1+\cos(2\omega t)]\mathcal{X}_{01}
-
\sin(2\omega t)\mathcal{Y}_{01}
\Bigr)
\nonumber\\
%%%%%%%%%%%%%%%%%%%%%%%%%%%%%%%%%%%%%%%%%%%%%%%%%
&{}
-\frac{\rmi}{2}g_2
(
\rme^{-\kappa t}
\mathcal{P}_\varphi
+
\rme^{\kappa t}
\mathcal{Q}_\varphi
)
\,\Bigl(
[1+\cos(2\omega t)]\mathcal{X}_{12}
-
\sin(2\omega t)\mathcal{Y}_{12}
\Bigr),
\end{align}
where  $\mathcal{Y}_{ij}\varrho=[Y_{ij},{}\varrho{}]$ with $Y_{ij}=\rmi\ket{i}\bra{j}-\rmi\ket{j}\bra{i}$ for $i,j=0,1,2$.

Notice that this generator in the reference frame is unbounded for $\kappa\to+\infty$. 
On the other hand, the integral action $\hat{\mathcal{S}}_{\omega,\kappa}(t)$ defined as~\eqref{eqn:S12} is bounded uniformly for $\omega,\kappa>0$, and
\begingroup
\allowdisplaybreaks
\begin{align}
\hat{\mathcal{S}}_{\omega,\kappa}(t)
={}&
{-\frac{\rmi}{2}}
\rme^{-\rmi\omega t\mathcal{H}}
\,\Biggl\{
g_1t
(
\mathcal{P}_\varphi
+
\rme^{-\kappa t}
\mathcal{Q}_\varphi
)
\mathcal{X}_{01}
\nonumber\\
%%%%%%%%%%%%%%%%%%%%%%%%%%%%%%%%%%%%%%%%%%%%%%%%%
&\hphantom{{-\frac{\rmi}{2}}\rme^{-\rmi\omega t\mathcal{H}}\,\Biggl\{}
{}
+
g_2
\frac{1-\rme^{-\kappa t}}{\kappa}
\left[
\left(
1+\frac{\kappa^2}{\kappa^2+4\omega^2}
\right)
\mathcal{X}_{12}
-
\frac{2\kappa\omega}{\kappa^2+4\omega^2}
(
\mathcal{P}_\varphi
-
\mathcal{Q}_\varphi
)
\mathcal{Y}_{12}
\right]
\Biggr\}
\nonumber\\
%%%%%%%%%%%%%%%%%%%%%%%%%%%%%%%%%%%%%%%%%%%%%%%%%
&{}
-\frac{\rmi}{2}
\frac{
\sin(\omega t)
}{\omega}
\,\Biggl[
g_1
(
\mathcal{P}_\varphi
+
\rme^{-\kappa t}
\mathcal{Q}_\varphi
)
\mathcal{X}_{01}
+
g_2
\,\Biggl(
\frac{4\omega^2}{\kappa^2+4\omega^2}
(
\rme^{-\kappa t}
\mathcal{P}_\varphi
+
\mathcal{Q}_\varphi
)
\mathcal{X}_{12}
\nonumber\\
%%%%%%%%%%%%%%%%%%%%%%%%%%%%%%%%%%%%%%%%%%%%%%%%%
&\hphantom{{}
-\frac{\rmi}{2}
\frac{
\sin(\omega t)
}{\omega}
\,\Biggl[
g_1
(
\mathcal{P}_\varphi
+
\rme^{-\kappa t}
\mathcal{Q}_\varphi
)
\mathcal{X}_{01}
+
g_2
}
\,{}+
\frac{2\kappa\omega}{\kappa^2+4\omega^2}
(
\rme^{-\kappa t}
\mathcal{P}_\varphi
-
\mathcal{Q}_\varphi
)
\mathcal{Y}_{12}
\Biggr)
\Biggr]\,
\rme^{-\rmi\omega t\mathcal{H}}
\nonumber\\
\sim{}&
{-\frac{\rmi}{2}}
g_1t
\rme^{-\rmi\omega t\mathcal{H}}
\mathcal{P}_\varphi
\mathcal{X}_{01},
%%%%%%%%%%%%%%%%%%%%%%%%%%%%%%%%%%%%%%%%%%%%%%%%%
\end{align}
\endgroup
as $\omega,\kappa\to+\infty$.
Then, by choosing 
\begin{equation}\label{eq:DbarEx3}
\overline{\mathcal{D}}
=-\frac{\rmi}{2}g_1\mathcal{P}_\varphi\mathcal{X}_{01}
=-\frac{\rmi}{2}g_1\mathcal{X}_{01}\mathcal{P}_\varphi,
\end{equation}
the peripheral part of the action defined as~\eqref{eqn:CondThm} reads
\begin{align}
\mathcal{S}_{\omega,\kappa,\varphi}(t)
={}&
\rme^{t\mathcal{L}_{0,\omega,\kappa}}
\int_0^t\rmd s\,
[
\rme^{-s\mathcal{L}_{0,\omega,\kappa}}
\mathcal{P}_\varphi
\mathcal{D}_\omega(s)
\mathcal{P}_\varphi
\rme^{s\mathcal{L}_{0,\omega,\kappa}}
-
\overline{\mathcal{D}}
]
\nonumber\\
={}&
{-\frac{\rmi}{2}}
g_1
\frac{
\sin(\omega t)
}{\omega}
\mathcal{P}_\varphi
\mathcal{X}_{01}
\rme^{-\rmi\omega t\mathcal{H}},
%%%%%%%%%%%%%%%%%%%%%%%%%%%%%%%%%%%%%%%%%%%%%%%%%
\end{align}
and we have bounds
\begin{equation}
\|\mathcal{S}_{\omega,\kappa,\varphi}(t)\|_\diamond
\le\frac{|g_1\sin(\omega t)|}{\omega}
,\quad
\|\mathcal{D}_\omega(t)\|_\diamond
\le2\sqrt{g_1^2+g_2^2}\,|{\cos(\omega t)}|
,\quad
\|\overline{\mathcal{D}}\|_\diamond
=|g_1|,
\label{eqn:DiamondEx3}
\end{equation}
where we have used $\|\mathcal{P}_\varphi\|_\diamond=1$ and the unitary invariance of the diamond norm, and have bounded the adjoint action $\ad_H(\varrho)=[H,\varrho]$ by $\|{\ad_H}\|_\diamond\le2\|H\|_\infty$, noting $\|g_1X_{01}+g_2X_{12}\|_\infty=\sqrt{g_1^2+g_2^2}$.

Since $\mathcal{L}_{\omega,\kappa}(t)$ in~\eqref{eqn:Example4a}--\eqref{eqn:Example4b} is a generator of the GKLS form, the evolution it generates is CPTP, and its diamond norm is given by $\|\Lambda_{\omega,\kappa}(t)\|_\diamond=1$.
We can thus apply Theorem~\ref{rmk:PhysicalBound}\@.
Although $\overline{\mathcal{D}}$ in~\eqref{eq:DbarEx3} itself is not of the GKLS form, the evolution it generates is a contraction on the peripheral subspace with $\|\overline{\Lambda}(t)\mathcal{P}_\varphi\|_\diamond=1$, as claimed in Theorem~\ref{rmk:PhysicalBound}\@.
On the other hand, $\overline{\Lambda}(t)=\rme^{-\frac{\rmi}{2}g_1t\mathcal{X}_{01}}\mathcal{P}_\varphi+\mathcal{Q}_\varphi$ is bounded by $\|\overline{\Lambda}(t)\|_\diamond\le2$, since $\|\mathcal{Q}_\varphi\|_\diamond=1$ in the present model.
Therefore, using the bound~\eqref{eqn:BoundThm6} of Theorem~\ref{rmk:PhysicalBound}, we get the uniform bound on $[0,T]$,
\begin{equation}\label{eq:boundEx3uniform}
    \sup_{t\in[0,T]}\|\Lambda_{\omega,\kappa}(t)-\Lambda_{\mathrm{RWA}}(t)\|_\diamond\leq {2}\left(\frac{|g_1|}{\omega}+\frac{4}{\kappa}\sqrt{g_1^2+g_2^2}\right)\left[
1
+
T\left(
|g_1|+2\sqrt{g_1^2+g_2^2}
\right)
\right],
\end{equation}
where the effective evolution $\Lambda_\mathrm{RWA}(t)$ in the RWA is generated by
\begin{equation}\label{eq:LRWAEx3}
\mathcal{L}_\mathrm{RWA}(t)
{{}={}}\mathcal{L}_{0,\omega,\kappa}
-\frac{\rmi}{2}g_1[\cos(\omega t)\mathcal{X}_{01}+\sin(\omega t)\mathcal{Y}_{01}]{\mathcal{P}_\varphi}.
\end{equation}

The dissipation in~\eqref{eq:L0Ex3} gives rise to an exponential decay in $\mathcal{Q}_\varphi$. If one is not interested in this transient decay, the distance between the true evolution and the approximated evolution restricted to the peripheral subspace $\mathcal{P}_\varphi$ can be considered instead. In this situation, we can use the bound~\eqref{eqn:BoundThm2PPhys} of Theorem~\ref{rmk:PhysicalBound}, which provides 
\begin{align}
&\|
\Lambda_{\omega,\kappa}(t)
-
\Lambda_\mathrm{RWA}(t)
\mathcal{P}_\varphi
\|_\diamond
\nonumber\\
&\qquad
\le
\left(
\frac{|g_1|}{\omega}
+\frac{2}{\kappa}
\sqrt{g_1^2+g_2^2}
\right)
\Bigl[
1
+
t\,\Bigl(
|g_1|+2\sqrt{g_1^2+g_2^2}
\Bigr)
\Bigr]
+
\frac{2}{\kappa}
\sqrt{g_1^2+g_2^2}
+\rme^{-\kappa t}.
\label{eq:boundEx3}
\end{align}
Note that this bound does not imply uniform convergence in $[0,T]$ (which would include the transient regime). However, in this case, the error can be controlled uniformly in $[\tau,T]$ for $0<\tau<T$, as $\kappa\rightarrow+\infty$,
\begin{align}
&\sup_{t\in[\tau,T]}\|
\Lambda_{\omega,\kappa}(t)
-
\Lambda_\mathrm{RWA}(t)
\mathcal{P}_\varphi
\|_\diamond
\nonumber\\
&\qquad
\le
\left(
\frac{|g_1|}{\omega}
+\frac{2}{\kappa}
\sqrt{g_1^2+g_2^2}
\right)
\Bigl[
1
+T\,\Bigl(
|g_1|+2\sqrt{g_1^2+g_2^2}
\Bigr)
\Bigr]
+
\frac{2}{\kappa}
\sqrt{g_1^2+g_2^2}
+\rme^{-\kappa\tau}.
\label{eq:boundEx3Puniform}
\end{align}
See Fig.~\ref{fig:Example3}\@.
\end{example}
\begin{figure}
    \centering
    \includegraphics[width=\linewidth]{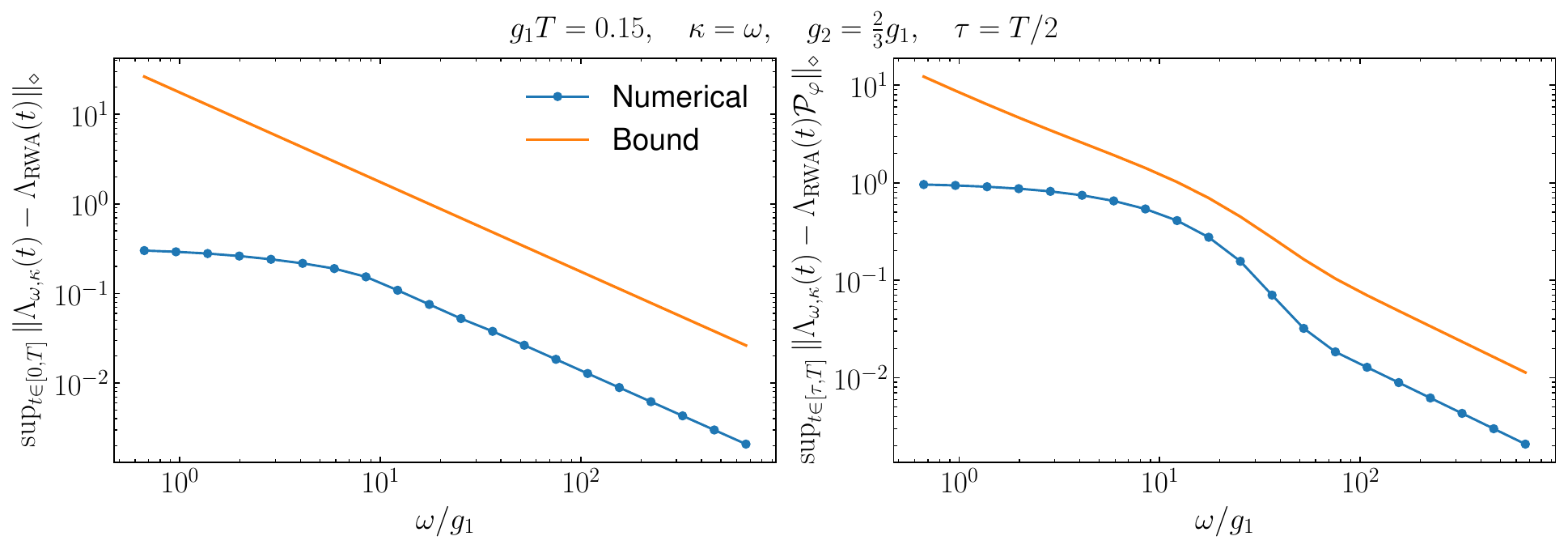}
    \caption{Comparison of the exact diamond distances computed numerically and the corresponding bounds. In the left panel the distance~\eqref{eq:boundEx3uniform} between the true evolution $\Lambda_{\omega,\kappa}(t)$ and the approximate one $\Lambda_\mathrm{RWA}(t)$ is shown, while in the right panel the distance~\eqref{eq:boundEx3Puniform} to the approximate evolution projected onto the peripheral subspace, $\Lambda_\mathrm{RWA}(t)\mathcal{P}_\varphi$, is shown. The knee at $\omega/g_1\simeq10$ marks the crossover to the fast-driving regime, $\omega T\gtrsim 1$, where oscillatory averaging becomes effective and the $1/\omega$ scaling sets in.}
    \label{fig:Example3}
\end{figure}

\section{Secular Approximation in the Redfield Equation}
\label{sec:SecApp}
The Redfield equation was first derived as an equation of motion for the density matrix of a system in contact with a thermal bath, in the context of  nuclear magnetic resonance~\cite{redfield1957theory}. It also appears as an intermediate result in the microscopic derivation of the GKLS master equation from the total Hamiltonian $H_{\mathrm{tot}}=H\otimes \bm{1}_B + \bm{1}_S \otimes  H_B+\lambda H_I$ of the system and its environment in the so-called weak-coupling limit~\cite{vanKampen,breuer2002theory,de2017dynamics}. In particular, after the Born-Markov approximation, the evolution operator is generated by
	the Redfield generator~\cite{redfield1957theory,vanKampen,chruscinski2022dynamical}
	\begin{equation}\label{eq:Redfield}
		\mathcal{L}_{\kappa}=\kappa \mathcal{L}_0+ \mathcal{D}, \qquad \mathcal{L}_0\varrho= -\mathrm{i}[H,{}\varrho{}], \qquad \mathcal{D}\varrho=\sum_{\alpha}[S_\alpha,{}\varrho E_\alpha^\dagger-E_\alpha\varrho{}],
	\end{equation}
	where $H=H^\dagger$, $S_\alpha=S_\alpha^\dagger$ for all $\alpha$, and $E_\alpha$ are related to $S_\alpha$ via the correlation functions $C_{\alpha\beta}(\tau)$ of the environment as
	\begin{equation}\label{eq:Ealpha}
		E_\alpha = \sum_{\beta} \int_0^\infty \dd \tau\,C_{\alpha\beta}(\tau)
        \ee^{\tau\mathcal{L}_0} S_\beta 
        = \sum_{\beta} \int_0^\infty \dd \tau\,C_{\alpha\beta}(\tau)\ee^{-\rmi\tau H} S_\beta \ee^{\rmi\tau H}.
	\end{equation} 
    Here $\kappa=\lambda^{-2}\gg 1$, where $\lambda$ is the (small) system-environment interaction strength, and we are considering the generator of the dynamics in the rescaled (macroscopic) time $\tau=\lambda^2 t$, with $t$ being the microscopic time~\cite{vH1955,vanhove99}.

    Notice that in general the Redfield generator $\mathcal{L}_\kappa$ is not of the GKLS form, and thus does not generate a \textit{bona fide} CPTP evolution. The secular approximation consists in decomposing the operators $S_\alpha$ and $E_\alpha$ appearing in the dissipator $\mathcal{D}$ according to the spectral projections of $\mathcal{L}_0$, and keeping only the diagonal terms. More concretely, let us write the spectral representation of $\mathcal{L}_0$ in terms of the spectral decomposition of
    \begin{equation}
		H=\sum_m\varepsilon_mP_m
	\end{equation} 
    as
    \begin{equation}\label{eq:specL0sec}
	\mathcal{L}_0=-\rmi\sum_j\omega_j\mathcal{P}_j, \qquad \mathcal{P}_j\varrho=\sum_{m,n}\delta_{\omega_j,\varepsilon_m-\varepsilon_n}P_{m}\varrho P_{n}.
	\end{equation}
    We introduce the notation $A(\omega_j)\equiv\mathcal{P}_j(A)$ for operator $A$, which will be useful in the following. When frequency resolution and Hermitian conjugation are combined, we keep the order of the two operations explicit: $A^\dagger(\omega_j)\equiv \mathcal{P}_j(A^\dagger)$ is the $\omega_j$-th component of $A^\dagger$, whereas $A(\omega_j)^\dagger\equiv[\mathcal{P}_j(A)]^\dagger$ denotes the Hermitian conjugate of the $\omega_j$-th component of $A$. Thus, $A(\omega_j)^\dagger=A^\dagger(-\omega_j)$. In particular, since $S_\alpha=S_\alpha^\dagger$, one has $S_\alpha(\omega_j)^\dagger=S_\alpha(-\omega_j)$.
    Using~\eqref{eq:specL0sec} in~\eqref{eq:Ealpha}, we have
	\begin{equation}\label{eq:EalphaOmega}
		E_{\alpha}=\sum_{\beta,k}\Gamma_{\alpha\beta}(\omega_k)S_{\beta}(\omega_k), \qquad \Gamma_{\alpha\beta}(\omega_k):=\int_0^\infty \dd \tau\, C_{\alpha\beta}(\tau)\ee^{-\ii \omega_{k}\tau}.
	\end{equation}
    The secular approximation consists in replacing the Redfield generator~\eqref{eq:Redfield} with 
    \begin{equation}
        \mathcal{L}_{\mathrm{sec}} = \kappa\mathcal{L}_0+\mathcal{D}_{\mathrm{sec}},
    \end{equation}
    where
    \begin{equation}\label{eq:GKLS}
        \mathcal{D}_\mathrm{sec}\varrho
		=-\ii[H_{LS},{}\varrho{}]
		-\frac{1}{2}\sum_{j}\sum_{\alpha,\beta}\gamma_{\alpha\beta}(\omega_j)\,\Bigl(
		\{{S_\alpha(\omega_j)^\dagger} S_\beta(\omega_j),{}\varrho{}\}
		-2S_\beta(\omega_j)\varrho {S_\alpha(\omega_j)^\dagger}
		\Bigr),
	\end{equation}
    \begin{equation}
		H_{LS}=\sum_{j} \sum_{\alpha,\beta}\sigma_{\alpha\beta}(\omega_j){S_\alpha(\omega_j)^\dagger} S_\beta(\omega_j),
	\end{equation}
    and 
    \begin{equation}\label{eq:gamma_sigma}
    \gamma_{\alpha\beta}(\omega_j)= \Gamma_{\alpha\beta}(\omega_j)+\Gamma_{\beta\alpha}^{*}(\omega_j), \qquad \sigma_{\alpha\beta}(\omega_j)=\frac{1}{2\rmi}[\Gamma_{\alpha\beta}(\omega_j)-\Gamma_{\beta\alpha}^*(\omega_j)].
\end{equation}
    $H_{LS}$ is the so-called Lamb-shift contribution. 
    It is known that $(\gamma_{\alpha\beta}(\omega_j))$ as a matrix is positive, and hence, the generator $\mathcal{D}_\mathrm{sec}$ in~\eqref{eq:GKLS} is a GKLS generator~\cite{alicki2007quantum,chruscinski2022dynamical,breuer2002theory}.

    We will now show that the approximation of the Redfield generator~\eqref{eq:Redfield} with the GKLS generator~\eqref{eq:GKLS} can be justified in our framework as an application of Corollary~\ref{cor:StrongLimit}\@.
    Moreover, since the generator~\eqref{eq:GKLS} is of the GKLS form and $\mathcal{L}_0$ in~\eqref{eq:Redfield} is the generator of a unitary evolution, we can use Theorem~\ref{rmk:PhysicalBoundUnitary} to bound the error in the approximation, using~\eqref{eqn:BoundThm2phys}. 
    First, note that in this case $\mathcal{L}_0$ is completely peripheral, so $R/\gamma_0=0$. Then,  $\mathcal{D}_Z$ in~\eqref{eq:DZ} specializes to
    \begin{align}
        \mathcal{D}_Z\varrho&=\sum_{\omega_j}\mathcal{P}_j\mathcal{D}\mathcal{P}_j\varrho\nonumber\\
        &=\sum_{\alpha}\sum_{\omega_j}\sum_{m,n}\sum_{\ell,k} \delta_{\omega_j,\varepsilon_{m}-\varepsilon_{n}} \delta_{\omega_j,\varepsilon_{\ell}-\varepsilon_{k}}P_\ell[S_\alpha,{}    P_{m}\varrho P_n E_\alpha^\dagger-E_\alpha    P_{m}\varrho P_n{}]P_k.
        \label{eq:bigCommutator}
    \end{align}
    The terms in~\eqref{eq:bigCommutator} arising from the commutator, where $S_\alpha$ and $E_\alpha$ are on different sides of $\varrho$, can be simplified as
\begingroup
\allowdisplaybreaks
    \begin{align}
    &\sum_{\alpha}\sum_{\omega_j}\sum_{m,n}\sum_{\ell,k} \delta_{\omega_j,\varepsilon_{m}-\varepsilon_{n}} \delta_{\omega_j,\varepsilon_{\ell}-\varepsilon_{k}}
    (P_\ell S_\alpha     P_{m}\varrho P_n E_\alpha^\dagger P_k+ P_\ell E_\alpha     P_{m}\varrho P_n S_\alpha P_k)\nonumber\\
    &\qquad =\sum_{\alpha}\sum_{\omega_j}[S_\alpha(\omega_j)\varrho E_\alpha^\dagger(-\omega_j)+E_\alpha(\omega_j)\varrho S_\alpha(-\omega_j)]\nonumber\\
    &\qquad =\sum_{\alpha}\sum_{\omega_j}[S_\alpha(\omega_j)\varrho E_\alpha(\omega_j)^\dagger+E_\alpha(\omega_j)\varrho S_\alpha(\omega_j)^\dagger],
    \label{eq:DZsec1}
    \end{align}
\endgroup
    where we used~\eqref{eq:specL0sec} and the fact that the double constraint forces $\varepsilon_\ell-\varepsilon_m=\varepsilon_k-\varepsilon_n$.
    The terms of~\eqref{eq:bigCommutator}, where $S_\alpha$ and $E_\alpha$ are both on one side of $\varrho$, lead to
     \begin{align}
    &-\sum_{\alpha}\sum_{\omega_j}\sum_{m,n}\sum_{\ell,k} \delta_{\omega_j,\varepsilon_{m}-\varepsilon_{n}} \delta_{\omega_j,\varepsilon_{\ell}-\varepsilon_{k}}
    (P_\ell S_\alpha E_\alpha    P_{m}\varrho P_n P_k+P_\ell P_m\varrho P_nE_\alpha^\dagger S_\alpha P_k)\nonumber\\
    &\qquad =-\sum_{\alpha}\sum_{\omega_j}\sum_{\ell,m,n}\delta_{\omega_j,\varepsilon_{m}-\varepsilon_{n}} \delta_{\omega_j,\varepsilon_{\ell}-\varepsilon_{n}}
    (P_\ell S_\alpha E_\alpha    P_{m}\varrho P_n+
    P_n \varrho P_\ell E_\alpha^\dagger S_\alpha P_m)\nonumber\\ 
    &\qquad =-\sum_{\alpha}\sum_{m,n}
    (P_m S_\alpha E_\alpha    P_{m}\varrho P_n+
    P_n \varrho P_m E_\alpha^\dagger S_\alpha P_m)
    \nonumber\\
    &\qquad =-\sum_{\alpha}\sum_{m}
    (P_m S_\alpha E_\alpha    P_{m}\varrho+
   \varrho P_m E_\alpha^\dagger S_\alpha P_m)
   \nonumber\\
    &\qquad =-\sum_{\alpha}\sum_{\omega_j}
    [S_\alpha(\omega_j) E_\alpha(-\omega_j)    \varrho+
   \varrho  E_\alpha^\dagger(\omega_j) S_\alpha(-\omega_j)]\nonumber\\
    &\qquad =-\sum_{\alpha}\sum_{\omega_j}
    [S_\alpha(\omega_j) E_\alpha(-\omega_j)    \varrho+
   \varrho  E_\alpha(-\omega_j)^\dagger S_\alpha(-\omega_j)].
    \label{eq:DZsec2}
    \end{align}
    Collecting~\eqref{eq:DZsec1} and~\eqref{eq:DZsec2}, and using~\eqref{eq:EalphaOmega}, Eq.~\eqref{eq:bigCommutator} translates into
	\begingroup
    \allowdisplaybreaks
    \begin{align}
    \mathcal{D}_Z\varrho&=\sum_{j}\sum_{\alpha,\beta}\Bigl(
		\Gamma_{\alpha\beta}^*(\omega_j)[S_\alpha(\omega_j)\varrho {S_\beta(\omega_j)^\dagger}-{}\varrho {S_\beta(\omega_j)^\dagger}S_\alpha(\omega_j)]\nonumber\\
		&\hphantom{{}=\sum_{j}\sum_{\alpha,\beta}\Bigl(}
		{}+\Gamma_{\alpha\beta}(-\omega_j)[S_{\beta}(-\omega_j)\varrho S_\alpha(\omega_j)-S_\alpha(\omega_j)S_\beta(-\omega_j)\varrho{}]
		\Bigr)\nonumber\\
		&=\sum_{j}\sum_{\alpha,\beta}\Bigl(
		\Gamma_{\beta\alpha}^*(\omega_j)[S_\beta(\omega_j)\varrho {S_\alpha(\omega_j)^\dagger}-{}\varrho {S_\alpha(\omega_j)^\dagger} S_\beta(\omega_j)]\nonumber\\
		&\hphantom{{}=\sum_{j}\sum_{\alpha,\beta}\Bigl(}
		{}+\Gamma_{\alpha\beta}(\omega_j)[S_{\beta}(\omega_j)\varrho{S_\alpha(\omega_j)^\dagger}-{S_\alpha(\omega_j)^\dagger}S_\beta(\omega_j)\varrho{}]\Bigr).
		\label{eq:aveDsecular}
	\end{align}	
	\endgroup
    The last equality follows from the fact that, for each $\ii\omega_j$, $-\ii\omega_j$ is also in the spectrum of $\mathcal{L}_0$, and $S_\alpha(-\omega_j)={S_\alpha(\omega_j)^\dagger}$. Finally, using the definitions in~\eqref{eq:gamma_sigma}, it is immediate to verify that $\mathcal{D}_Z=\mathcal{D}_\mathrm{sec}$ from~\eqref{eq:GKLS}.

The strong part $\kappa\mathcal{L}_0$ of the Redfield generator~\eqref{eq:Redfield} is the generator of a unitary evolution, and the generator $\mathcal{D}_\mathrm{sec}=\mathcal{D}_Z$ in~\eqref{eq:GKLS} after the secular approximation is of the GKLS form, which generates a CPTP and therefore contractive evolution in the diamond norm. Therefore we can apply Theorem~\ref{rmk:PhysicalBoundUnitary} with the diamond norm and $\overline{L}=1$, or with any other norm which makes $\rme^{t \mathcal{L}_0}$ an isometry. In the latter case,
$\overline L$ may depend on the dimension and the chosen norm, through norm-equivalence constants, but can be chosen independently of $T$ and $\kappa$.
The error in the secular approximation can then be bounded uniformly for $t\in[0,T]$ using~\eqref{eqn:BoundThm2phys},
\begin{equation}
		\|\rme^{t(\kappa \mathcal{L}_0+\mathcal{D})}-\rme^{t\mathcal{L}_{\mathrm{sec}}}\|\leq \overline{L}\left[\frac{1}{1-\|\mathcal{S}_\kappa\|_{\infty,T}}
\exp\!\left(
\frac{
\|\mathcal{S}_\kappa\|_{\infty,T}
}{1-\|\mathcal{S}_\kappa\|_{\infty,T}}
\overline{L}(\|\mathcal{D}\|+\|\mathcal{D}_Z\|)T
\right)
-1\right],
\label{eqn:SecularAppBound}
\end{equation}
for $\|\mathcal{S}_\kappa\|_{\infty,T}<1$, where the integral action $\mathcal{S}_\kappa(t)$ is bounded by
\begin{equation}
\|\mathcal{S}_\kappa\|_{\infty,T}
\le
\frac{2m(m-1)\|\mathcal{D}\|}{\kappa\eta}.
\label{eq:BoundActionSecular}
\end{equation}
See \eqref{eqn:SStrongCoupling} of Corollary~\ref{cor:StrongLimit}\@.
Note that $P=1$, since the spectral projections of a unitary are bounded by $1$\@.

Even though the bound~\eqref{eqn:SecularAppBound} is exponential in $T$, we now show that, since 
$\|\mathcal{D}\|$, $\|\mathcal{D}_Z\|$, and $\overline{L}$ do not depend on $T$, it can be linearized so that it grows at most linearly in $T$ up to intervals of length $T=O(\kappa)$.
Let $S=\|\mathcal{S}_\kappa\|_{\infty,T}$ and $M=\|\mathcal{D}\|+\|\mathcal{D}_Z\|$.
Since 
\begin{equation}
\log\!\left(\frac{1}{1-S}\right)=\int_0^S\frac{\rmd x}{1-x}\le\frac{S}{1-S},   
\end{equation}
the prefactor of the exponential in~\eqref{eqn:SecularAppBound} can be absorbed into the exponent and bounded as
\begin{equation}
\exp\!\left[
  \log\!\left(\frac{1}{1-S}\right)+\frac{S}{1-S}\overline{L}MT
  \right]
  \le \rme^{\chi},
  \quad
  \chi=\frac{(1+\overline{L}MT)S}{1-S}.
\end{equation}
Then, note that for
\begin{equation}\label{eq:ConditionsLinearization}
S\le\frac{1}{4},\qquad
T\le\frac{1}{2S\overline{L}M},
\end{equation}
the argument $\chi$ of the exponential is bounded by
\begin{equation}
\chi\le\frac{4}{3}(1+\overline{L}MT)S\le1,
\end{equation}
and we have 
\begin{equation}\label{eq:Linearization}
\rme^{\chi}-1
\le (\rme-1)\chi
\le\frac{4}{3}(\rme-1)(1+\overline{L}MT)S.
\end{equation}
Then, using~\eqref{eq:ConditionsLinearization}, \eqref{eq:Linearization} together with~\eqref{eq:BoundActionSecular} and the fact that $\frac{4}{3}(\rme-1)\leq2.3$, the bound \eqref{eqn:SecularAppBound} is linearized to
\begin{equation}
\|\rme^{t(\kappa \mathcal{L}_0+\mathcal{D})}-\rme^{t\mathcal{L}_{\mathrm{sec}}}\|
\le\overline{L}\frac{4.6m(m-1)\|\mathcal{D}\|}{\kappa\eta}[1+\overline{L}(\|\mathcal{D}\|+\|\mathcal{D}_Z\|)T],
\label{eqn:SecularAppBoundLinear}
\end{equation}
for
\begin{equation}
\kappa\ge\frac{8m(m-1)\|\mathcal{D}\|}{\eta},\qquad
T\le\frac{\kappa\eta}{4m(m-1)\overline{L}\|\mathcal{D}\|(\|\mathcal{D}\|+\|\mathcal{D}_Z\|)}.
\label{eqn:SecularAppBoundLinearConds}
\end{equation}
This shows that, for sufficiently large $\kappa$, the error bound grows at most linearly in $T$ over a time interval whose length grows linearly with $\kappa$.

We stress that the estimates~\eqref{eqn:SecularAppBound}--\eqref{eq:BoundActionSecular} and~\eqref{eqn:SecularAppBoundLinear}--\eqref{eqn:SecularAppBoundLinearConds} are valid for finite times and do not provide a uniform bound for $T\to+\infty$\@.
It would be interesting to compare the asymptotic stationary states of the Redfield master equation and of the master equation in the secular approximation, but it would require additional spectral information not assumed in the present framework.

\section{Conclusions}
We have developed a theoretical framework for the derivation of nonperturbative error bounds in the approximation of the evolution operator in open quantum systems. 
The framework can be easily adapted to concrete situations and allows one to control the error by bounding the integral action of the difference between the generators of the exact and approximated dynamics in a suitable reference frame. 
We applied the framework to justify the RWA, showcasing some concrete qubit and qutrit examples. 
The same reference frame used to compute the action allows one to establish whether the noise should be changed in the approximation. 
In particular, if the generator of the rotating frame commutes with the noise, it is not affected by the approximation, as shown in  Example~\ref{ex:2}. 
If, on the contrary, they do not commute, the noise should be replaced by its long-time-averaged version in the rotating frame, as showcased in  Example~\ref{ex:3}. 
The framework developed is also suitable for dealing with situations where the strong part of the dynamics giving rise to the highly oscillatory terms also contains  dissipation, such as in Example~\ref{ex:4}\@. Finally, we showed that the bounds obtained also allow  rigorous control of the error in the secular approximation, when the Redfield generator is replaced with a GKLS generator.

Possible future directions of this work include extending the framework to the infinite-dimensional case and considering possibly unbounded operators. Such an extension would be useful for example in the derivation of error bounds in the secular approximation before tracing out the environment~\cite{fleming2010rotating}. An additional potential extension would be to consider iterated integration-by-parts techniques to improve on the long-time validity of the bound, in the same vein as Ref.~\cite{dey2025error}. 
A further direction would be to complement the present worst-case diamond-norm estimates with state-dependent error bounds for prescribed initial states, which may provide tighter estimates in specific applications, and could be particularly useful in the infinite-dimensional setting~\cite{burgarth2024taming}. Finally, a natural extension of the present framework would be to derive error bounds for partial approximations, by introducing an appropriate multiscale separation between frequency clusters and retaining the slowly varying couplings within each cluster.

\section*{Acknowledgments}
DB acknowledges discussions with Robin Hillier.  
This work was supported in part by the Top Global University Project from the Ministry of Education, Culture, Sports, Science and Technology (MEXT), Japan.
KY was supported by JSPS KAKENHI Grant No.~JP24K06904 from the Japan Society for the Promotion of Science (JSPS)\@.
GG acknowledges financial support from PNRR MUR Project PE0000023-NQSTI and from the project ``Patto territoriale sistema universitario pugliese'' (PATTI, CUP F61B23000370006)\@.
GG and PF acknowledge support from INFN through the project ``QUANTUM'', from the Italian National Group of Mathematical Physics (GNFM-INdAM), and from the Italian funding within the ``Budget MUR - Dipartimenti di Eccellenza 2023--2027''  - Quantum Sensing and Modelling for One-Health (QuaSiModO).

\paragraph{Author contribution.}
The authors contributed equally to this work. 
\paragraph{AI tools usage.} The authors used AI tools (ChatGPT and Claude)  during the preparation of this manuscript to assist with proofreading and language editing of the manuscript. All resulting text was reviewed and edited by the authors, who take full responsibility for the content of the work.

\appendix
\section{Diamond Norm}
\label{app:DiamondNorm}
\begin{definition}[Diamond norm~\cite{watrous2018theory}]
The diamond norm (completely bounded trace norm) $\|\Phi\|_\diamond$ of a linear map $\Phi$ on operators is defined by
\begin{equation}
\|\Phi\|_\diamond=\|\Phi\otimes1\|_1=\sup_{\|A\|_1=1}\|(\Phi\otimes1)(A)\|_1,
\end{equation}
where $\|A\|_1$ is the trace norm for operators, and the identity $1$ of the extended map $\Phi\otimes1$ acts on the space of the same dimension as the one on which the map $\Phi$ acts.
\end{definition}
\begin{definition}[Choi operator~\cite{watrous2018theory}]
Let $\Phi$ be a linear map on $B(\mathbb{C}^d)$, the Banach space of operators on $\mathbb{C}^d$, and take an orthonormal basis $\{\ket{e_i}\}_{i=1,\ldots,d}$. Then, the Choi representation (Choi operator) of the map $\Phi$ is defined by
\begin{equation}
C(\Phi)=\frac{1}{d}(\Phi\otimes1)(|\openone)(\openone|),
\end{equation}
where
\begin{equation}
|\openone)=\sum_{i=1}^d\ket{e_i}\otimes\ket{e_i}.
\label{eqn:vecOne}
\end{equation}
\end{definition}
\begin{prop}[Choi matrix~\cite{havel2003robust,wood2011tensor}]\label{prop:ChoiMatrix}
Let $\Phi$ be a linear map on $B(\mathbb{C}^d)$, and consider its expansion
\begin{equation}
\Phi(A)=\sum_{\mu,\nu=0}^{d^2-1}
c_{\mu\nu}F_\mu AF_\nu^\dag,
\label{eqn:MapExpansion}
\end{equation}
in terms of an orthonormal basis of operators $\{F_\mu\}_{\mu=0,1,\ldots,d^2-1}$ satisfying
\begin{equation}
\Tr(F_\mu^\dag F_\nu)=\delta_{\mu\nu}\qquad(\mu,\nu=0,1,\ldots d^2-1).
\end{equation}
Then, the $d^2\times d^2$ matrix $c=(c_{\mu\nu})$ of the coefficients of the expansion of $\Phi$ in~\eqref{eqn:MapExpansion} is a matrix representation of the operator $dC(\Phi)$, i.e., of the Choi operator $C(\Phi)$ of $\Phi$ multiplied by $d$.
\end{prop}
\begin{proof}
Notice first that the vectors
\begin{equation}
|F_\mu)=(F_\mu\otimes\openone)|\openone)\qquad(\mu=0,1,\ldots,d^2-1)
\end{equation}
defined on the vector $|\openone)$ in~\eqref{eqn:vecOne} form a complete set of orthonormal vectors satisfying
\begin{equation}
(F_\mu|F_\nu)
=(\openone|F_\mu^\dag F_\nu\otimes\openone|\openone)
=\Tr(F_\mu^\dag F_\nu)
=\delta_{\mu\nu}\qquad(\mu,\nu=0,1,\ldots,d^2-1).
\end{equation}
Notice also that the Choi operator $C(\Phi)$ of the map $\Phi$ is given by
\begin{equation}
C(\Phi)
=\frac{1}{d}\sum_{\mu,\nu=0}^{d^2-1}c_{\mu\nu}|F_\mu)(F_\nu|.
\end{equation}
Then, the matrix representation of the Choi operator $C(\Phi)$ on the basis $\{|F_\mu)\}_{\mu=0,1,\ldots,d^2-1}$ is given by
\begin{equation}
(F_\mu|C(\Phi)|F_\nu)
=\frac{1}{d}c_{\mu\nu}\qquad(\mu,\nu=0,1,\ldots,d^2-1).
\end{equation}
\end{proof}

\begin{prop}[The diamond norm and Choi operator~\cite{watrous2018theory,hahn2022unification}]\label{prop:DiamondBound}
Let $\Phi$ be a linear map on $B(\mathbb{C}^d)$ and let $C(\Phi)$ be its Choi operator.
Then,
\begin{equation}
\|C(\Phi)\|_1\le\|\Phi\|_\diamond\le d\|C(\Phi)\|_1.
\label{eqn:DiamondChoiBounds}
\end{equation}
\end{prop}

\begin{prop}[Diamond norm of a class of qubit maps]\label{prop:DiamondBoundQubit}
Let $\Phi$ be a linear map on $B(\mathbb{C}^2)$, the
operators of a qubit, and $\Phi^*$ be its dual map, satisfying the conditions
\begin{equation}
[\Phi(A)]^\dag=\Phi(A^\dag),\qquad
\Phi(\openone)=\Phi^*(\openone)=\alpha\openone,
\label{eqn:SubunitalCond}
\end{equation}
for some real constant $\alpha$.
Then,
\begin{equation}
\|\Phi\|_\diamond=\|C(\Phi)\|_1,
\label{eqn:DiamondNormQubit}
\end{equation}
where $C(\Phi)$ is the Choi operator of $\Phi$.
\end{prop}
\begin{proof}
Let us take
\begin{equation}
(F_0,F_1,F_2,F_3)
=\left(
\frac{1}{\sqrt{2}}\openone,
\frac{1}{\sqrt{2}}X,
\frac{1}{\sqrt{2}}Y,
\frac{1}{\sqrt{2}}Z
\right),
\label{eqn:PauliBasisF}
\end{equation}
as a complete set of orthonormal basis operators, with $X$, $Y$, and $Z$ being Pauli operators, and consider the expansion~\eqref{eqn:MapExpansion} of $\Phi$.
The conditions in~\eqref{eqn:SubunitalCond} impose the following structure on the coefficient matrix $c=(c_{\mu\nu})$ of the expansion~\eqref{eqn:MapExpansion}:
\begin{equation}
c
=\begin{pmatrix}
a_{00}&\rmi h_1&\rmi h_2&\rmi h_3\\
-\rmi h_1&a_{11}&a_{12}&a_{13}\\
-\rmi h_2&a_{12}&a_{22}&a_{23}\\
-\rmi h_3&a_{13}&a_{23}&a_{33}
\end{pmatrix},
\end{equation}
with real numbers $h_1,h_2,h_3,a_{00},a_{11},a_{12},a_{13},a_{22},a_{23},a_{33}\in\mathbb{R}$, and $\Tr c=a_{00}+a_{11}+a_{22}+a_{33}=2\alpha$.
This matrix $c$ can be diagonalized $U^\dag cU=\diag(\lambda_0,\lambda_1,\lambda_2,\lambda_3)$ by a unitary matrix of the form $U=VR$, where $V=\diag(\rmi,1,1,1)$ is a diagonal unitary and $R$ is a real orthogonal matrix satisfying $R^T=R^{-1}$.
Then, the expansion~\eqref{eqn:MapExpansion} is simplified to
\begin{equation}
\Phi(A)=\sum_{\mu=0}^3
\lambda_\mu G_\mu AG_\mu^\dag,
\qquad
G_\mu=\sum_{\nu=0}^3U_{\nu\mu}F_\nu.
\end{equation}
By using the facts that $U$ is unitary and that $U_{0\mu}\in\rmi\mathbb{R}$, $U_{i\mu}\in\mathbb{R}$, for $i=1,2,3$ and $\mu=0,1,2,3$, one has
\begin{equation}
G_\mu^\dag G_\mu
=
\frac{1}{2}
\sum_{\nu}U_{\nu\mu}^*U_{\nu\mu}\openone
+
\frac{1}{\sqrt{2}}
\sum_{k=1}^3\left(
2\Re(U_{0\mu}^*U_{k\mu})
-\sum_{i,j=1}^3\varepsilon_{ijk}\Im(U_{i\mu}^*U_{j\mu})
\right)
F_k
=\frac{1}{2}\openone,
\end{equation}
and hence,
\begin{equation}
\|G_\mu\|_\infty=\frac{1}{\sqrt{2}}.	
\end{equation}
Therefore,
\begingroup
\allowdisplaybreaks
\begin{align}
\|(\Phi\otimes1)(A)\|_1
&=\left\|
\sum_{\mu=0}^3
\lambda_\mu
(G_\mu\otimes\openone)A(G_\mu^\dag\otimes\openone)
\right\|_1
\nonumber\\
&\le\sum_{\mu=0}^3
|\lambda_\mu|
\|(G_\mu\otimes\openone)A(G_\mu^\dag\otimes\openone)\|_1
\nonumber
\displaybreak[0]
\\
&\le\sum_{\mu=0}^3
|\lambda_\mu|
\|G_\mu\otimes\openone\|_\infty\|A\|_1\|G_\mu^\dag\otimes\openone\|_\infty
\nonumber
\displaybreak[0]
\\
&=\frac{1}{2}
\sum_{\mu=0}^3
|\lambda_\mu|
\|A\|_1
\nonumber
\displaybreak[0]
\\
&=\frac{1}{2}
\|c\|_1
\|A\|_1
\vphantom{\sum_{\mu=0}^3}
\nonumber
\displaybreak[0]
\\
&=
\|C(\Phi)\|_1
\|A\|_1
\vphantom{\frac{1}{2}},
\end{align}
\endgroup
and
\begin{equation}
\|\Phi\|_\diamond\le\|C(\Phi)\|_1,
\end{equation}
where we have used $\|ABC\|\le\|A\|_\infty\|B\|\|C\|_\infty$ for any unitarily invariant norm~\cite[Proposition~IV.2.4]{bhatia2013matrix}, $\|G_\mu\otimes\openone\|_\infty=\|G_\mu\|_\infty$, and Proposition~\ref{prop:ChoiMatrix}.
This together with the lower bound proven in Proposition~\ref{prop:DiamondBound} implies the equality~\eqref{eqn:DiamondNormQubit}.
\end{proof}

\begin{example}
\label{ex:QubitMaps}
Let us look at a few examples relevant to the models analyzed in Examples~\ref{ex:2}--\ref{ex:4}:
\begin{align}
\Phi_1(\varrho)
&=
-\rmi h
[
X
,\varrho
]
-\frac{1}{2}\gamma(\varrho-Z\varrho Z),
\\
\Phi_2(\varrho)
&=
-\rmi h
[
X,\varrho
]
-\frac{1}{2}\gamma(\varrho-X\varrho X),
\\
\Phi_3(\varrho)
&=
-\rmi h
[
X
,{}\varrho{}
]
-\frac{1}{4}\gamma(2{}\varrho{}-X{}\varrho{}X-Y{}\varrho{}Y),
\\
\Phi_4(\varrho)
&=
\frac{1}{2}\gamma(X{}\varrho{}X-Y{}\varrho{}Y),
\end{align}
with $h\in\mathbb{R}$ and $\gamma\ge0$\@.
These are qubit maps, and we use Proposition~\ref{prop:DiamondBoundQubit} to compute their diamond norms.
The conditions in~\eqref{eqn:SubunitalCond} of Proposition~\ref{prop:DiamondBoundQubit} are fulfilled, with $\alpha=0$.
We use the basis~\eqref{eqn:PauliBasisF} to construct the Choi matrices below.

Let us start with $\Phi_1$.
Its Choi matrix reads
\begin{equation}
C(\Phi_1)
=
\begin{pmatrix}
-\gamma/2&\rmi h&0&0\\
-\rmi h&0&0&0\\
0&0&0&0\\
0&0&0&\gamma/2
\end{pmatrix}.
\end{equation}
Since its eigenvalues are 
\begin{equation}
\left\{
\frac{1}{4}\,\Bigl(-\gamma\pm\sqrt{\gamma^2+16h^2}\Bigr),\ %
\frac{1}{2}\gamma,\ 0
\right\},
\end{equation}
its trace norm yields the diamond norm
\begin{equation}
\|\Phi_1\|_\diamond
=
\|C(\Phi_1)\|_1
=
\frac{1}{2}\,\Bigl(
\gamma
+\sqrt{\gamma^2+16h^2}
\Bigr).
\label{eqn:DiamondPhi1}
\end{equation}

Let us next look at $\Phi_2$.
Its Choi matrix reads
\begin{equation}
C(\Phi_2)
=
\begin{pmatrix}
-\gamma/2&\rmi h&0&0\\
-\rmi h&\gamma/2&0&0\\
0&0&0&0\\
0&0&0&0
\end{pmatrix}.
\end{equation}
Since its eigenvalues are 
\begin{equation}
\left\{
\pm\frac{1}{2}\sqrt{\gamma^2+4h^2},\ %
0,\ 0
\right\},
\end{equation}
its trace norm yields the diamond norm
\begin{equation}
\|\Phi_2\|_\diamond
=
\|C(\Phi_2)\|_1
=
\sqrt{\gamma^2+4h^2}.
\label{eqn:DiamondPhi2}
\end{equation}

For the third example $\Phi_3$, its Choi matrix reads
\begin{equation}
C(\Phi_3)
=
\begin{pmatrix}
-\gamma/2&\rmi h&0&0\\
-\rmi h&\gamma/4&0&0\\
0&0&\gamma/4&0\\
0&0&0&0
\end{pmatrix}.
\end{equation}
Since its eigenvalues are 
\begin{equation}
\left\{
\frac{1}{8}\,\Bigl(-\gamma\pm\sqrt{9\gamma^2+64h^2}\Bigr),\ %
\frac{1}{4}\gamma,\ 0
\right\},
\end{equation}
its trace norm yields the diamond norm
\begin{equation}
\|\Phi_3\|_\diamond
=
\|C(\Phi_3)\|_1
=
\frac{1}{4}\,\Bigl(\gamma+\sqrt{9\gamma^2+64h^2}\Bigr).
\label{eqn:DiamondPhi3}
\end{equation}

Finally, let us look at $\Phi_4$.
Its Choi matrix reads
\begin{equation}
C(\Phi_4)
=
\begin{pmatrix}
0&0&0&0\\
0&\gamma/2&0&0\\
0&0&-\gamma/2&0\\
0&0&0&0
\end{pmatrix}.
\end{equation}
Since its eigenvalues are 
\begin{equation}
\left\{
\pm\frac{1}{2}\gamma,\ 0,\ 0
\right\},
\end{equation}
its trace norm yields the diamond norm
\begin{equation}
\|\Phi_4\|_\diamond
=
\|C(\Phi_4)\|_1
=
\gamma.
\label{eqn:DiamondPhi4}
\end{equation}
\end{example}

\begin{example}\label{ex:NormQ}
Let us consider a pair of nonzero mutually orthogonal projections $P_1$ and $P_2$ in $\mathbb{C}^d$, with ranks $d_1=\Tr P_1$ and $d_2=\Tr P_2$ (which are restricted by $d_1+d_2\le d$), and let us define the map
\begin{equation}
\mathcal{Q}(A)=P_1AP_2+P_2AP_1.
\label{eqn:OffDiagMap}
\end{equation}
Its Choi operator is given by
\begin{equation}
C(\mathcal{Q})
=\frac{\sqrt{d_1d_2}}{d}
\,\Bigl(
\ket{\phi_1}\bra{\phi_2}
+
\ket{\phi_2}\bra{\phi_1}
\Bigr),
\end{equation}
with some normalized vectors $\ket{\phi_1}$ and $\ket{\phi_2}$ purifying $P_1/d_1$ and $P_2/d_2$, respectively, and orthogonal to each other, $\bracket{\phi_1}{\phi_2}=0$\@.
Its trace norm reads
\begin{equation}
\|C(\mathcal{Q})\|_1
=\frac{2\sqrt{d_1d_2}}{d}\le1.
\end{equation}
The equality holds iff $d_1=d_2=d/2$.
On the other hand, since $\mathcal{Q}$ is a projection, its diamond norm $\|\mathcal{Q}\|_\diamond$ fulfills $\|\mathcal{Q}\|_\diamond=\|\mathcal{Q}^2\|_\diamond\le\|\mathcal{Q}\|_\diamond^2$, and this implies $\|\mathcal{Q}\|_\diamond\ge1$.
In addition, by arranging~\eqref{eqn:OffDiagMap} as
\begin{equation}
\mathcal{Q}(A)
=\frac{1}{2}(P_1+P_2)A(P_1+P_2)
-
\frac{1}{2}(P_1-P_2)A(P_1-P_2),
\end{equation}
we have a bound
\begin{equation}
\|(\mathcal{Q}\otimes1)(A)\|_1
\le
\frac{1}{2}\|P_1+P_2\|_\infty^2\|A\|_1
+
\frac{1}{2}\|P_1-P_2\|_\infty^2\|A\|_1
=\|A\|_1,
\end{equation}
and get $\|\mathcal{Q}\|_\diamond\le1$.
Therefore, 
\begin{equation}
\|\mathcal{Q}\|_\diamond=1,
\end{equation}
and
\begin{equation}
\|C(\mathcal{Q})\|_1
\le\|\mathcal{Q}\|_\diamond=1.
\end{equation}
In particular, for $d_1\neq d_2$ or $d_1=d_2<d/2$, one has $\|C(\mathcal{Q})\|_1<\|\mathcal{Q}\|_\diamond$. 
\end{example}

\section{Useful Properties of Superoperators}
\label{app:ElementaryFacts}
In this Appendix, we recall some elementary facts on superoperators, which are used in the Examples of Section~\ref{sec:Examples}\@.
\begin{prop}\label{prop:UConj}
	Let $\mathcal{U}\varrho=U\varrho U^\dagger$, with $U$ a unitary operator, and let $\mathcal{A}\varrho=[A,{}\varrho{}]$. Then, $\mathcal{U}\mathcal{A}\mathcal{U}^{-1}\varrho=[U A U^\dagger,{}\varrho{}]$.
\end{prop}
\begin{proof}
	It follows from a simple computation,
	\begin{equation}
\mathcal{U}\mathcal{A}\mathcal{U}^{-1}\varrho=U[A,U^{\dagger}\varrho U]U^\dagger= {[UAU^\dagger,\varrho ]}.
	\end{equation}
\end{proof}
\begin{prop}\label{prop:blockdiagsuper}
	Let $\{P_j\}$ be a collection of projections such that $P_k P_\ell=\delta_{k\ell}P_k$ and $\sum_kP_k=\openone$, and define
	\begin{equation}
		\mathcal{P}\varrho=\sum_kP_k\varrho P_k,\qquad \mathcal{Q}\varrho=(1-\mathcal{P})\varrho=\sum_{k\neq\ell} P_k\varrho P_\ell.
	\end{equation}
	Then, the following properties hold.
	\begin{itemize}
		\item[(i)] Let $A=\sum_kP_kAP_k$ be a block-diagonal operator with respect to the decomposition $\{P_k\}$. Then, $\mathcal{A}\varrho=[A,{}\varrho{}]$ satisfies $[\mathcal{A},\mathcal{P}]=[\mathcal{A},\mathcal{Q}]=0$.
		\item[(ii)] If $B$ is off-diagonal with respect to the decomposition $\{P_k\}$, i.e.~$P_kBP_k=0$, then $\mathcal{B}\varrho=[B,{}\varrho{}]$ satisfies $\mathcal{PBP}=0$. If, in addition, there are only two blocks, i.e.~$k\in\{1,2\}$, then also $\mathcal{Q}\mathcal{B}\mathcal{Q}=0$.
	\end{itemize}
\end{prop}
\begin{proof}
    To see (i), note that $P_kA=P_kAP_k=AP_k$, which implies that
    \begin{align}
        \mathcal{PA}X&=\sum_{k}(P_kAXP_k-P_kXAP_k)\nonumber\\
    &=\sum_k( P_kAP_k X P_k-P_kX P_k AP_k)\nonumber\\
    &=\sum_k (AP_k X P_k-P_kX P_k A)=\mathcal{AP}X.
    \end{align}
    Since $\mathcal{Q}=1-\mathcal{P}$, it follows also that $[\mathcal{Q},\mathcal{A}]=0$.
    To see (ii), note that
   \begingroup
   \allowdisplaybreaks
   \begin{align}
    \mathcal{PBP}X&=\sum_{j,k}P_j(BP_kXP_k-P_kXP_kB)P_j\\
        &=\sum_{j}P_jBP_jXP_j-P_jXP_jBP_j=0.
    \end{align}
    \endgroup
    Finally
    \begin{align}
        \mathcal{QBQ}X&=\sum_{j\neq k}\sum_{\ell\neq m} P_j(BP_\ell XP_m-P_\ell XP_mB)P_k,
    \end{align}
    and it is immediate to verify that if there are only two blocks $P_1$ and $P_2$, all the terms in the summation contain necessarily either $P_1P_2=P_2P_1=0$ or $P_kBP_k=0$.
\end{proof}

%\bibliographystyle{prsty-title-hyperref}
%\bibliography{Refs}

\begin{thebibliography}{10}

\bibitem{bloch1940magnetic}
F. Bloch and A. Siegert, Magnetic Resonance for Nonrotating Fields, \href{https://doi.org/10.1103/PhysRev.57.522}{Phys. Rev. \textbf{57},  522  (1940)}.

\bibitem{shirley1965solution}
J.~H. Shirley, Solution of the {Schr\"odinger} Equation with a {Hamiltonian} Periodic in Time, \href{https://doi.org/10.1103/PhysRev.138.B979}{Phys. Rev. \textbf{138},  B979  (1965)}.

\bibitem{haeberlen1968coherent}
U. Haeberlen and J.~S. Waugh, Coherent Averaging Effects in Magnetic Resonance, \href{https://doi.org/10.1103/PhysRev.175.453}{Phys. Rev. \textbf{175},  453  (1968)}.

\bibitem{Agarwal}
G.~S. Agarwal, Rotating-Wave Approximation and Spontaneous Emission, \href{https://doi.org/10.1103/PhysRevA.7.1195}{Phys. Rev. A \textbf{7},  1195  (1973)}.

\bibitem{Nazir}
A. Stokes and A. Nazir, Gauge ambiguities imply {Jaynes}-{Cummings} physics remains valid in ultrastrong coupling {QED}, \href{https://doi.org/10.1038/s41467-018-08101-0}{Nat. Commun. \textbf{10},  499  (2019)}.

\bibitem{burgarth2022one}
D. Burgarth, P. Facchi, G. Gramegna, and K. Yuasa, One bound to rule them all: from {Adiabatic} to {Zeno}, \href{https://doi.org/10.22331/q-2022-06-14-737}{Quantum \textbf{6},  737  (2022)}.

\bibitem{burgarth2024taming}
D. Burgarth, P. Facchi, R. Hillier, and M. Ligab{\`o}, Taming the Rotating Wave Approximation, \href{https://doi.org/10.22331/q-2024-02-21-1262}{Quantum \textbf{8},  1262  (2024)}.

\bibitem{richter2024quantifying}
L. Richter, D. Burgarth, and D. Lonigro, Quantifying the rotating-wave approximation of the {Dicke} model, \href{https://doi.org/10.1088/1751-8121/ae42a3}{J. Phys. A: Math. Theor. \textbf{59},  075203  (2026)}.

\bibitem{heib2025bounding}
T. Heib, P. Lageyre, A. Ferreri, F.~K. Wilhelm, G.~S. Paraoanu, D. Burgarth, A.~W. Schell, and D.~E. Bruschi, Bounding the rotating wave approximation for coupled harmonic oscillators, \href{https://doi.org/10.1088/1751-8121/adcd16}{J. Phys. A: Math. Theor. \textbf{58},  175304  (2025)}.

\bibitem{dey2025error}
A. Dey, D. Lonigro, K. Yuasa, and D. Burgarth, Error bounds for the {Floquet}-{Magnus} expansion and their application to the semiclassical quantum {Rabi} model, \href{https://doi.org/10.1103/6bgj-s987}{Phys. Rev. A \textbf{112},  053723  (2025)}.

\bibitem{lindblad1976generators}
G. Lindblad, On the generators of quantum dynamical semigroups, \href{https://doi.org/10.1007/BF01608499}{Commun. Math. Phys. \textbf{48},  119  (1976)}.

\bibitem{gorini1976completely}
V. Gorini, A. Kossakowski, and E.~C.~G. Sudarshan, Completely positive dynamical semigroups of {$N$}-level systems, \href{https://doi.org/10.1063/1.522979}{J. Math. Phys. \textbf{17},  821  (1976)}.

\bibitem{FloquetLindbladian}
A. Schnell, A. Eckardt, and S. Denisov, Is there a {Floquet} {Lindbladian}?, \href{https://doi.org/10.1103/PhysRevB.101.100301}{Phys. Rev. B \textbf{101},  100301(R)  (2020)}.

\bibitem{band2015open}
Y.~B. Band, Open quantum system stochastic dynamics with and without the {RWA}, \href{https://doi.org/10.1088/0953-4075/48/4/045401}{J. Phys. B: At. Mol. Opt. Phys. \textbf{48},  045401  (2015)}.

\bibitem{fleming2010rotating}
C. Fleming, N.~I. Cummings, C. Anastopoulos, and B.-L. Hu, The rotating-wave approximation: consistency and applicability from an open quantum system analysis, \href{https://doi.org/10.1088/1751-8113/43/40/405304}{J. Phys. A: Math. Theor. \textbf{43},  405304  (2010)}.

\bibitem{Kohler}
S. Kohler, T. Dittrich, and P. H{\"a}nggi, {Floquet}-{Markovian} description of the parametrically driven, dissipative harmonic quantum oscillator, \href{https://doi.org/10.1103/PhysRevE.55.300}{Phys. Rev. E \textbf{55},  300  (1997)}.

\bibitem{Hartmann}
R. Hartmann and W.~T. Strunz, Accuracy assessment of perturbative master equations: Embracing nonpositivity, \href{https://doi.org/10.1103/PhysRevA.101.012103}{Phys. Rev. A \textbf{101},  012103  (2020)}.

\bibitem{benatti2022local}
F. Benatti, D. Chru{\'s}ci{\'n}ski, and R. Floreanini, Local Generation of Entanglement with {Redfield} Dynamics, \href{https://doi.org/10.1142/S1230161222500019}{Open Syst. Inf. Dyn. \textbf{29},  2250001  (2022)}.

\bibitem{Trushechkin}
A. Trushechkin, Unified {Gorini}-{Kossakowski}-{Lindblad}-{Sudarshan} quantum master equation beyond the secular approximation, \href{https://doi.org/10.1103/PhysRevA.103.062226}{Phys. Rev. A \textbf{103},  062226  (2021)}.

\bibitem{alicki2007quantum}
R. Alicki and K. Lendi, \href{https://doi.org/10.1007/3-540-70861-8}{\textit{Quantum Dynamical Semigroups and Applications}, 2nd ed.} (Springer, Berlin, 2007).

\bibitem{nielsen2001quantum}
M.~A. Nielsen and I.~L. Chuang, \href{https://doi.org/10.1017/CBO9780511976667}{\textit{Quantum Computation and Quantum Information}, 10th anniversary ed.} (Cambridge University Press, Cambridge, 2010).

\bibitem{chruscinski2017brief}
D. Chru{\'s}ci{\'n}ski and S. Pascazio, A Brief History of the {GKLS} Equation, \href{https://doi.org/10.1142/S1230161217400017}{Open Syst. Inf. Dyn. \textbf{24},  1740001  (2017)}.

\bibitem{reed1975ii}
M. Reed and B. Simon, \href{https://www.elsevier.com/books/ii-fourier-analysis-self-adjointness/reed/978-0-08-092537-0}{\textit{Methods of Modern Mathematical Physics II: {Fourier} Analysis, Self-Adjointness}} (Academic Press, San Diego, 1975).

\bibitem{kato1953integration}
T. Kato, Integration of the equation of evolution in a {Banach} space, \href{https://doi.org/10.2969/jmsj/00520208}{J. Math. Soc. Jpn. \textbf{5},  208  (1953)}.

\bibitem{avron2012adiabatic}
J.~E. Avron, M. Fraas, G.~M. Graf, and P. Grech, Adiabatic Theorems for Generators of Contracting Evolutions, \href{https://doi.org/10.1007/s00220-012-1504-1}{Commun. Math. Phys. \textbf{314},  163  (2012)}.

\bibitem{watrous2018theory}
J. Watrous, \href{https://doi.org/10.1017/9781316848142}{\textit{The Theory of Quantum Information}} (Cambridge University Press, Cambridge, 2018).

\bibitem{chruscinski2022dynamical}
D. Chru{\'s}ci{\'n}ski, Dynamical maps beyond {Markovian} regime, \href{https://doi.org/10.1016/j.physrep.2022.09.003}{Phys. Rep. \textbf{992},  1  (2022)}.

\bibitem{kato2013perturbation}
T. Kato, \href{https://doi.org/10.1007/978-3-642-66282-9}{\textit{Perturbation Theory for Linear Operators}, 2nd ed.} (Springer, Berlin, 1995).

\bibitem{wolf2020quantum}
M.~M. Wolf, Quantum Channels \& Operations: Guided Tour, 2012, {URL:} \href{https://mediatum.ub.tum.de/node?id=1701036}{https://mediatum.ub.tum.de/node?id=1701036}.

\bibitem{hasenohrl2022interaction}
M. Hasen{\"o}hrl and M.~M. Wolf, ``Interaction-Free'' Channel Discrimination, \href{https://doi.org/10.1007/s00023-022-01175-z}{Ann. Henri Poincar\'e \textbf{23},  3331  (2022)}.

\bibitem{wolf2008dividing}
M.~M. Wolf and J.~I. Cirac, Dividing Quantum Channels, \href{https://doi.org/10.1007/s00220-008-0411-y}{Commun. Math. Phys. \textbf{279},  147  (2008)}.

\bibitem{wolf2010inverse}
M.~M. Wolf and D. Perez-Garcia, The inverse eigenvalue problem for quantum channels, \href{https://doi.org/10.48550/arXiv.1005.4545}{arXiv:1005.4545 [quant-ph]  (2010)}.

\bibitem{burgarth2019generalized}
D. Burgarth, P. Facchi, H. Nakazato, S. Pascazio, and K. Yuasa, Generalized Adiabatic Theorem and Strong-Coupling Limits, \href{https://doi.org/10.22331/q-2019-06-12-152}{Quantum \textbf{3},  152  (2019)}.

\bibitem{Tao2006Nonlinear}
T. Tao, \href{https://doi.org/10.1090/cbms/106}{\textit{Nonlinear Dispersive Equations: Local and Global Analysis}} (American Mathematical Society, Providence, 2006).

\bibitem{gong2020error}
Z. Gong, N. Yoshioka, N. Shibata, and R. Hamazaki, Error bounds for constrained dynamics in gapped quantum systems: Rigorous results and generalizations, \href{https://doi.org/10.1103/PhysRevA.101.052122}{Phys. Rev. A \textbf{101},  052122  (2020)}.

\bibitem{burgarth2021eternal}
D. Burgarth, P. Facchi, H. Nakazato, S. Pascazio, and K. Yuasa, Eternal adiabaticity in quantum evolution, \href{https://doi.org/10.1103/PhysRevA.103.032214}{Phys. Rev. A \textbf{103},  032214  (2021)}.

\bibitem{Trushechkin2022}
A. Trushechkin, Quantum master equations and steady states for the ultrastrong-coupling limit and the strong-decoherence limit, \href{https://doi.org/10.1103/PhysRevA.106.042209}{Phys. Rev. A \textbf{106},  042209  (2022)}.

\bibitem{MarcantoniMerkli2025}
S. Marcantoni and M. Merkli, Ultrastrong coupling, nonselective measurement and quantum {Zeno} dynamics, \href{https://doi.org/10.22331/q-2025-03-06-1656}{Quantum \textbf{9},  1656  (2025)}.

\bibitem{redfield1957theory}
A.~G. Redfield, On the Theory of Relaxation Processes, \href{https://doi.org/10.1147/rd.11.0019}{IBM J. Res. Dev. \textbf{1},  19  (1957)}.

\bibitem{vanKampen}
N.~G. Van~Kampen, \href{https://doi.org/10.1016/B978-0-444-52965-7.X5000-4}{\textit{Stochastic Processes in Physics and Chemistry}, 3rd ed.} (Elsevier, Amsterdam, 2007).

\bibitem{breuer2002theory}
H.-P. Breuer and F. Petruccione, \href{https://doi.org/10.1093/acprof:oso/9780199213900.001.0001}{\textit{The Theory of Open Quantum Systems}} (Oxford University Press, Oxford, 2002).

\bibitem{de2017dynamics}
I. de~Vega and D. Alonso, Dynamics of non-{Markovian} open quantum systems, \href{https://doi.org/10.1103/RevModPhys.89.015001}{Rev. Mod. Phys. \textbf{89},  015001  (2017)}.

\bibitem{vH1955}
L. van Hove, Energy corrections and persistent perturbation effects in continuous spectra, \href{https://doi.org/10.1016/S0031-8914(55)92832-9}{Physica \textbf{21},  901  (1955)}.

\bibitem{vanhove99}
P. Facchi and S. Pascazio, Deviations from exponential law and {Van Hove's} ``$\lambda^2t$'' limit, \href{https://doi.org/10.1016/S0378-4371(99)00209-5}{Physica A \textbf{271},  133  (1999)}.

\bibitem{havel2003robust}
T.~F. Havel, Robust procedures for converting among {Lindblad}, {Kraus} and matrix representations of quantum dynamical semigroups, \href{https://doi.org/10.1063/1.1518555}{J. Math. Phys. \textbf{44},  534  (2003)}.

\bibitem{wood2011tensor}
C.~J. Wood, J.~D. Biamonte, and D.~G. Cory, Tensor networks and graphical calculus for open quantum systems, \href{https://doi.org/10.26421/QIC15.9-10-3}{Quantum Inf. Comput. \textbf{15},  0759  (2015)}.

\bibitem{hahn2022unification}
A. Hahn, D. Burgarth, and K. Yuasa, Unification of random dynamical decoupling and the quantum {Zeno} effect, \href{https://doi.org/10.1088/1367-2630/ac6b4f}{New J. Phys. \textbf{24},  063027  (2022)}.

\bibitem{bhatia2013matrix}
R. Bhatia, \href{https://doi.org/10.1007/978-1-4612-0653-8}{\textit{Matrix Analysis}} (Springer, New York, 1997).

\end{thebibliography}
%\end{document}

\end{document}